\documentclass[fullpage,11pt]{article}

\usepackage[sc]{mathpazo}
\usepackage{fullpage,amsthm}
\usepackage{graphicx} % support the \includegraphics command and options
\usepackage{array} % for better arrays (eg matrices) in maths
\usepackage{amsmath, amssymb, amsfonts, verbatim}
\usepackage{hyphenat,epsfig,subfigure,multirow}
\usepackage[usenames,dvipsnames]{xcolor}

\newcommand{\EAT}[1]{}

\usepackage{tcolorbox}
\definecolor{DarkRed}{rgb}{0.5,0.1,0.1}
\definecolor{DarkBlue}{rgb}{0.1,0.1,0.5}

\usepackage{hyperref}
\hypersetup{
colorlinks=true,
pdfnewwindow=true,
citecolor=ForestGreen,
linkcolor=DarkRed,
filecolor=DarkRed,
urlcolor=DarkBlue
}

\usepackage{bm}
\usepackage{url}
\usepackage{xspace} 
\usepackage[mathscr]{euscript}
\usepackage{algorithm}
\usepackage[noend]{algpseudocode}
\makeatletter
\def\BState{\State\hskip-\ALG@thistlm}
\makeatother

\usepackage{cite}
\usepackage{enumerate}

\usepackage[a4paper,margin=1in]{geometry}

\newtheorem{lemma}{Lemma}[section]
\newtheorem{theorem}[lemma]{Theorem}

\newtheorem{corollary}[lemma]{Corollary}
\newtheorem{claim}[lemma]{Claim}

\newtheorem{definition}{Definition}
\newtheorem{problem}{Problem}
\newtheorem{remark}[lemma]{Remark}
\newtheorem*{claim*}{Claim}
\newtheorem*{problem*}{Problem}

\allowdisplaybreaks

\renewcommand{\qed}{\nobreak \ifvmode \relax \else
      \ifdim\lastskip<1.5em \hskip-\lastskip
      \hskip1.5em plus0em minus0.5em \fi \nobreak
      \vrule height0.75em width0.5em depth0.25em\fi}

\usepackage[T1]{fontenc}
\usepackage[utf8]{inputenc}

\newcommand{\ourinfo}[1]{Department of Computer and Information Science, University of Pennsylvania. Supported in part by National Science
  Foundation grants CCF-1116961, CCF-1552909, and IIS-1447470 and an Adobe research award.
\newline\noindent Email: \texttt{#1}.}
\newenvironment{tbox}{\begin{tcolorbox}[
		enlarge top by=5pt,
		enlarge bottom by=5pt,
		 boxsep=0pt,
                  left=4pt,
                  right=4pt,
                  top=10pt,
                  arc=0pt,
                  boxrule=1pt,toprule=1pt,
                  colback=white
                  ]%%
	}
{\end{tcolorbox}}

\renewenvironment{proof}[1][Proof]{\paragraph{#1.}\xspace}{\hfill$\qed$}
\newcommand{\ssection}[1]{\noindent\textbf{#1}}
\newcommand{\toShrink}{-.25cm}
\newcommand{\toShrinkEnu}{-.2cm}
\newcommand{\toShrinkEqn}{-.45cm}

\newenvironment{cLemma}{\vspace{\toShrink}\begin{lemma}}{\end{lemma}\vspace{\toShrink}}

\newenvironment{cEqnarray}{\vspace{\toShrink}\begin{eqnarray*}}{\end{eqnarray*}\vspace{\toShrinkEqn}\\}

\newenvironment{cEnumerate}[1][1.]{\vspace{\toShrinkEnu}\begin{enumerate}[#1]
\setlength{\itemsep}{0pt}
\setlength{\parskip}{0pt}
\setlength{\parsep}{0pt}}
{\end{enumerate}\vspace{\toShrinkEnu}}

\newcommand{\toShrinkTextbox}{0cm}
\newcommand{\cTextbox}[2]{\vspace{\toShrinkTextbox}\textbox{#1}{\vspace{\toShrinkTextbox} #2 \vspace{\toShrinkTextbox}}}

\renewcommand{\bar}[1]{\overline{#1}}

\newcommand{\etal}{{\it et al.\,}}
\newcommand{\eps}{\epsilon}

\newcommand{\card}[1]{\left\vert{#1}\right\vert}
\newcommand{\IR}{\ensuremath{\mathbb{R}}}

\newcommand{\ceil}[1]{{\left\lceil{#1}\right\rceil}}
\newcommand{\floor}[1]{{\left\lfloor{#1}\right\rfloor}}

\newcommand{\set}[1]{\ensuremath{\left\{ #1 \right\}}}

\newcommand{\poly}{\mbox{\rm poly}}

\newcommand{\REM}[1]{}
\newcommand{\Ot}{\widetilde{O}}

\newcommand{\union}{\ensuremath{\cup}}

\newcommand{\textbox}[2]{
{
\begin{tbox}
\textbf{#1}
{#2}
\end{tbox}
}
}

\newcommand{\eat}[1]{}

\newcommand{\bv}[1]{\textbf{#1}}

\newcommand{\bx}{\bv{x}}

\newcommand{\app}[2]{#1|_{#2}}

\newcommand{\dtlg}{\ensuremath{\,:\!-\;}}

\newcommand{\fullcover}{\text{\sf Coverage}\xspace}
\newcommand{\cq}{\text{CQ$^{\neg}$}}

\newcommand{\UCQ}{\text{\normalfont UCQ}\xspace}

\newcommand{\ucqs}{\text{UCQ$^{\neg}$s}}

\newcommand{\qnsub}{\ensuremath{Q_{\tt NOTSUB}}}

\newcommand{\FC}{\mathcal{F}}

\newcommand{\qexp}{\ensuremath{Q_{\tt EXP}}}
\newcommand{\norm}[1]{\ensuremath{\|#1\|}}
\newcommand{\rowsample}{\mathrm{RowSample}}
\newcommand{\regressionSketch}{{\sc REG-Sketch}\xspace}

\newcommand{\Setdis}{{\small {\sf SetDisjointness}}\xspace}
\newcommand{\queryFont}[1]{{\sf #1}}
\newcommand{\regQuery}{\queryFont{regression}\xspace}
\newcommand{\countQuery}{\queryFont{count}\xspace}
\newcommand{\sumQuery}{\queryFont{sum}\xspace}
\newcommand{\avgQuery}{\queryFont{average}\xspace}

\newcommand{\quantileQuery}{\queryFont{quantiles}\xspace}
\newcommand{\stconQuery}{\queryFont{st-connectivity}\xspace}
\newcommand{\lrceil}[1]{\ensuremath{ \lceil #1 \rceil}}
\newcommand{\lrangle}[1]{\ensuremath{ \langle #1 \rangle}}
\newcommand{\row}[2]{\ensuremath{ #1_{(#2)} }}

\newcommand{\bvA}{\ensuremath{\bv{A}}}
\newcommand{\bvv}{\ensuremath{\bv{v}}}
\newcommand{\bva}{\ensuremath{\bv{a}}}
\newcommand{\bvM}{\ensuremath{\bv{M}}}
\newcommand{\bvb}{\ensuremath{\bv{b}}}
\newcommand{\bvx}{\ensuremath{\bv{x}}}
\newcommand{\bve}{\ensuremath{\bv{e}}}
\newcommand{\bvp}{\ensuremath{\mathcal{P}}}

\newcommand{\probP}{\ensuremath{\bvp=(p_1, p_2, \dots, p_n)}}
\newcommand{\probPi}{\ensuremath{\bvp_i=(p_{i,1}, p_{i,2}, \dots, p_{i,n})}}

\newcommand{\countSketch}{{\sc CNT-Sketch}\xspace}
\newcommand{\sumSketch}{{\sc SUM-Sketch}\xspace}
\newcommand{\quantilesSketch}{{\sc QTL-Sketch}\xspace}
\newcommand{\trail}[1]{\text{\sc trail}\ensuremath{(#1)}}
\newcommand{\thtrail}[2]{{\sc \lrangle{#1, #2}-trail}}
\newcommand{\hash}{\ensuremath{g}}

\newcommand{\eiT}{\ensuremath{e_i^T}}
\newcommand{\edScheme}{$(\eps, \delta)$-linear provisioning scheme\xspace}
\newcommand{\edSchemes}{$(\eps, \delta)$-linear provisioning schemes\xspace}

\newcommand{\record}{\textbf{record}\xspace}

\newcommand{\compose}{\ensuremath{\langle Q_{L}; G_{\bar{A}} ; Q_{N} \rangle}\xspace}
\title{Algorithms for Provisioning Queries and Analytics}
\author{Sepehr Assadi\thanks{\ourinfo{\{sassadi,sanjeev,yangli2\}@cis.upenn.edu}}\and
Sanjeev Khanna\footnotemark[1] \and
Yang Li\footnotemark[1] \and
Val Tannen\footnote{Department of Computer and Information Science, University of Pennsylvania. Supported in part by National Science Foundation grants IIS-1217798 and IIS-1302212.
Email: \texttt{val@cis.upenn.edu}.}
}

\date{}

\renewcommand{\hat}[1]{\widehat{#1}}
\renewcommand{\tilde}[1]{\widetilde{#1}}
\begin{document}

\maketitle

\thispagestyle{empty}
\begin{abstract}
Provisioning is a technique for avoiding repeated expensive
computations in \emph{what-if analysis}.  Given a query, an analyst
formulates $k$ \emph{hypotheticals}, each retaining some of the tuples
of a database instance, {\em possibly overlapping}, and she wishes to
answer the query under \emph{scenarios}, where a scenario is defined
by a subset of the hypotheticals that are ``turned on''.  We say that
a query admits \emph{compact} provisioning if given any database instance and
any $k$ hypotheticals, one can create a poly-size (in $k$) \emph{sketch}
that can then be used to answer the query under any of the $2^{k}$
possible scenarios without accessing the original instance.

In this paper, we focus on provisioning complex queries that combine
relational algebra (the logical component), grouping, and 
statistics/analytics (the numerical component). 
We first show that queries that compute
quantiles or linear regression (as well as simpler
queries that compute count and sum/average of positive values)
can be compactly provisioned to provide (multiplicative) \emph{approximate} 
answers to an arbitrary precision. 
In contrast, \emph{exact} provisioning for each of these statistics
requires the sketch size to be exponential in $k$.
We then establish that for any complex query whose
logical component is a \emph{positive} relational algebra query, as long as
the numerical component can be compactly provisioned, the complex
query itself can be compactly provisioned. 
On the other hand, introducing
negation \emph{or} recursion in the logical component again requires the sketch
size to be exponential in $k$.
While our positive results use algorithms that do not access the original
instance after a scenario is known, we prove our lower bounds even for
the case when, knowing the scenario, limited access to the instance is allowed.
\end{abstract}

\clearpage
\setcounter{page}{1}

\tableofcontents
\newpage

\section{Introduction}
\label{sec:intro}

``What if analysis'' is a common technique for investigating the
impact of decisions on outcomes in science or business. It almost
always involves a data analytics computation. Nowadays such a
computation typically processes very large amounts of data and thus
may be expensive to perform, especially repeatedly. An analyst is
interested in exploring the computational impact of multiple
\emph{scenarios} that assume modifications of the input to the
analysis problem. Our general aim is to avoid repeating expensive
computations for each scenario. For a given problem, and starting from
a given set of potential scenarios, we wish to perform just \emph{one}
possibly expensive computation producing a small \emph{sketch} (i.e., a compressed representation of the input) such
that the answer for any of the given scenarios can be derived rapidly
from the sketch, without accessing the original (typically very large)
input.  We say that the sketch is ``provisioned'' to deal with the
problem under any of the scenarios and
following~\cite{deutch2013caravan}, we call the whole approach
\emph{provisioning}. Again, the goal of provisioning is to allow an
analyst to efficiently explore a multitude of scenarios, using only
the sketch and thus avoiding expensive recomputations for each
scenario.

In this paper, we apply the provisioning approach to queries that perform
\emph{in-database analytics}~\cite{HellersteinRSWFGNWFLK12}%
\footnote{In practice, the MADlib project~\cite{MADlib} has been one
of the pioneers for in-database analytics, primarily in collaboration
with Greenplum DB~\cite{GpDB}.
By now, major RDBMS products such as IBM DB2, MS SQL Server, and Oracle DB
already offer the ability to combine extensive analytics with
SQL queries.
}.
These are queries that combine
\emph{logical} components (relational algebra and Datalog), grouping,
and \emph{numerical} components (e.g., aggregates, 
quantiles and linear 
regression). Other analytics are discussed under further work.

Abstracting away any data integration/federation,
we will assume that the inputs are
relational instances and that the scenarios are defined by a set of
\emph{hypotheticals}. We further assume that each hypothetical
indicates the fact that certain tuples of an input instance are
\emph{retained}
(other semantics for hypotheticals are discussed under further work).

A scenario consists of turning on/off each of the hypotheticals.
Applying a scenario to an input instance therefore means keeping only
the tuples retained by at least one of the hypotheticals that are
turned on.  Thus, a trivial sketch can be obtained by applying each
scenario to the input, solving the problem for each such modified
input and collecting the answers into the sketch. However, with $k$
hypotheticals, there are \emph{exponentially} (in $k$) many scenarios.
Hence, even with a moderate number of
hypotheticals, the size of the sketch could be enormous.  Therefore,
as part of the statement of our problem we will aim to provision a
query by an algorithm that maps each (large) input instance to a \emph{compact}
(essentially size $\poly(k)$) sketch.

\paragraph{Example.} Suppose a large retailer has many and diverse
sales venues (e.g., its own stores, its own web site, through
multiple other stores, and through multiple other web retailers).
An analyst working for the retailer
is interested in learning, for each product in, say, ``Electronics'',
a regression model for the way in which the \emph{revenue} from the
product depends on both a sales venue's \emph{reputation}
(assume a numerical score) and a sales venue \emph{commission} (in \%; 0\%
if own store).
Moreover, the analyst wants to ignore
products with small sales volume unless they have a large MSRP (manufacturer's
suggested retail price).
Usually there is a large (possibly distributed/federated)
database that captures enough information to allow the computation
of such an analytic query.
For simplicity we assume in this example that the revenue
for each product ID and each sales venue is in one table
and thus we have the following query with a self-explanatory schema:
\begin{verbatim}
SELECT x.ProdID, LIN_REG(x.Revenue, z.Reputation, z.Commission) AS (B, A1, A2)
FROM  RevenueByProductAndVenue x
INNER JOIN Products y ON x.ProdID=y.ProdID
INNER JOIN SalesVenues z ON x.VenueID=z.VenueID
WHERE y.ProdCategory="Electronics" AND (x.Volume>100 OR y.MSRP>1000)
GROUP BY x.ProdID
\end{verbatim}
The syntax for treating linear regression as a multiple-column-aggregate
is simplified for illustration purposes in this example.
Here the values under the attributes \texttt{B,A1,A2} denote, for
each \texttt{ProdID}, the coefficients of the linear
regression model that is learned, i.e.,
\texttt{Revenue = B + A1*Reputation + A2*Commission}.

A desirable what-if analysis for this query may involve
hypotheticals such as retaining certain venue types, retaining certain
venues with specific sales tax properties, retaining certain product types
(within the specified category, e.g., tablets), and many others.
Each of these hypotheticals can in fact be implemented
as selections on one or more of the tables in the query
(assuming that the schema includes the appropriate information).
However, combining hypotheticals into scenarios is problematic.
The hypotheticals overlap and thus cannot be separated.
With $10$ (say) hypotheticals there will be $2^{10} = 1024$ (in
practice at least hundreds) of regression models of interest
for each product. Performing a lengthy computation for each one of these
models is in total very onerous.  Instead, we can \emph{provision}
the what-if analysis of this query
since the query in this example falls within the class covered by our
positive results.

\paragraph{Our results.} Our goal is to characterize the feasibility of provisioning with sketches of
\emph{compact} size (see Section~\ref{sec:prob} for a formal
definition) for a practical class of
\emph{complex queries} that consist of a \emph{logical}
component (relational algebra or Datalog), followed by a \emph{grouping}
component, and then by a \emph{numerical} component (aggregate/analytic)
that is applied to each group (a more detailed definition is given
in Section~\ref{sec:mix}).

The main challenge that we address, and the part where our main
contribution lies, is the design of compact provisioning schemes for
numerical queries, specifically linear ($\ell_2$) regression and quantiles.
Together with the usual count, sum and average, these are defined in
Section~\ref{sec:num-qry} as queries that take a set of numbers or of
tuples as input and return a number or a tuple of constant width as
output. It turns out that if we expect exact answers, then none of
these queries can be compactly provisioned.  However, we show that
compact provisioning schemes indeed exist for all of them if we relax
the objective to computing near-exact answers (see
Section~\ref{sec:prob} for a formal definition). The following theorem
summarizes our results for numerical queries (see
Section~\ref{sec:num-qry}):

\begin{theorem}[Informal]\label{thm:numerical}
  The \emph{quantiles, linear ($\ell_2$) regression, count, and
    sum/average (of positive numbers) queries} can be compactly
  provisioned to provide (multiplicative) approximate answers to an arbitrary
  precision, while their exact provisioning requires the sketch size
  to be exponential in the number of hypotheticals.
\end{theorem}

Our results on provisioning numerical queries can then be used
for complex queries, as the following theorem summarizes
(see Section~\ref{sec:mix}):

\begin{theorem}[Informal]\label{thm:complex}
Any complex query whose logical component is a \emph{positive
  relational algebra} query can be compactly provisioned to provide an
approximate answer to an arbitrary precision as long as its numerical
component can be compactly provisioned for the same precision, and as
long as the number of groups is not too large.  On the other hand,
introducing \emph{negation} or \emph{recursion} in the logical
component requires the sketch size to be
exponential in the number of hypotheticals.
\end{theorem}

%%\begin{remark}\label{rem:examination}
%%Our positive results describe algorithms that do \emph{not} access the original
%%instance when answering a given scenario. However, we prove our lower
%%bounds for the case when \emph{some} hindsight could be exploited, namely,
%%when extracting out of the sketch the answer for a given scenario, algorithms
%%can revisit $o(\card{I})$ tuples of the original instance $I$ and these tuples can be chosen knowing the scenario (see Theorem~\ref{thm:full-cover}). 
%%\end{remark}

\paragraph{Our techniques.}
At a high-level, our approach for compact provisioning can be described as follows. We start by building a sub-sketch for each hypothetical by focusing solely on the retained tuples of each hypothetical individually.
We then examine these sub-sketches against each other and collect additional information from the original input to summarize the effect
of appearance of other hypotheticals to each already computed sub-sketch. The first step usually involves using well-known (and properly adjusted) sampling or sketching techniques, while the second step,
which is where we concentrate the bulk of our efforts, is responsible for gathering the information required for combining the sketches and specifically dealing with overlapping hypotheticals.
Given a scenario, we answer the query by fetching the corresponding sub-sketches and merging them together; the result is a new sketch that act as sketch for the input consist of the \emph{union} of the hypotheticals.

We prove our lower bounds by first identifying a central problem,
i.e., the \fullcover problem (see Problem~\ref{prob:full-cover}), with
provably large space requirement for any provisioning scheme, and then
we use reductions from this problem to establish lower bounds for
other queries of interest. The space requirement of the \fullcover
problem itself is proven using simple tools from information theory
(see Theorem~\ref{thm:full-cover}).

\paragraph{Comparison with existing work.}
Our techniques for compact provisioning share some similarities with
those used in data streaming and in the distributed computation models
of~\cite{cormode2010optimal,cormode2011algorithms,woodruff2012tight} 
(see Section~\ref{sec:dis} for further discussion and formal separations),
and in particular with \emph{linear sketching}, which corresponds to
applying a linear transformation to the input data to obtain the
sketch.  However, due to overlap in the input, our sketches are
required to to be composable with the \emph{union} operation (instead
of the \emph{addition} operation obtained by linear sketches) and
hence linear sketching techniques are not directly applicable.

Dealing with duplicates in the input (similar to the overlapping hypotheticals) has also been considered in streaming and distributed computation models (see, e.g.,~\cite{cormode2005space,CormodeMuthukrishnanZhuang06}), which consider sketches that are ``duplicate-resilient''. 
Indeed, for simple queries like count, a direct application of these sketches is
sufficient for compact provisioning (see Section~\ref{sec:csa}). We also remark that the Count-Min sketch~\cite{cormode2005improved} can be applied to approximate quantiles even in the presence of duplication (see~\cite{CormodeMuthukrishnanZhuang06}), i.e., is duplicate-resilient. However, the approximation
guarantee achieved by the Count-Min sketch for quantiles is only \emph{additive} (i.e., $\pm \eps n$), in contrast to the stronger notion of \emph{multiplicative} approximation (i.e., $(1 \pm \eps)$) that we achieve in this paper.  
To the best of our knowledge, there is no similar result concerning duplicate-resilient sketches for multiplicative approximation of quantiles or the linear regression problem, and existing techniques do not
seem to be applicable for our purpose. Indeed one of the primary technical contributions of this paper is
designing provisioning schemes that can effectively deal with overlapping hypotheticals for quantiles and linear regression.

\paragraph{Further related work.} \emph{Provisioning}, in the sense used in this paper, originated
in~\cite{deutch2013caravan} together with a proposal for how to
perform it taking advantage of provenance tracking. Answering queries under
hypothetical updates is studied 
in~\cite{GhandeharizadehHJ96,BalminPP00} but the focus there is on
using a specialized warehouse to avoid transactional costs. (See
also~\cite{deutch2013caravan} for more related work.)

Estimating the number of distinct elements (corresponding to the count query) 
has been studied extensively in data streams~\cite{flajolet1985probabilistic,alon1996space,bar2002counting,kane2010optimal} and in certain distributed
computation models
%distributed functional monitoring
~\cite{cormode2010optimal,cormode2011algorithms,woodruff2012tight}. 
For estimating quantiles,
%in the data stream or the distributed model,
~\cite{manku1998approximate,gilbert2002summarize,greenwald2001space,cormode2005improved,greenwald2004power,yi2013optimal}
achieve an additive error of $\eps n$ for an input of length $n$,
and~\cite{gupta2003counting,cormode2006space} achieve a (stronger
guarantee of) $(1\pm \eps)$-approximation. 
Sampling and sketching techniques for $\ell_2$-regression problem have also been studied in~\cite{drineas2006sampling,sarlos2006improved,drineas2011faster,clarkson2009numerical}
for either speeding up the computation or in data streams (see~\cite{mahoney2011randomized,woodruff2014sketching} for excellent 
surveys on this topic).

%%\paragraph{Organization.} The rest of the paper is organized as follows. We start by introducing the provisioning model formally and defining the necessary notation for the rest
%%of the paper in Section~\ref{sec:prob}. In Section~\ref{sec:num-qry}, we provide our results for numerical queries and prove Theorem~\ref{thm:numerical}. Section~\ref{sec:mix} contains 
%%our results for the complex queries and the proof of Theorem~\ref{thm:complex}. We provide a comparison between the techniques for query provisioning and the distributed computation model
%%in Section~\ref{sec:dis}. Section~\ref{sec:full-cover} is dedicated to the proof of our main lower bound result, i.e., the lower bound for the \fullcover problem. Finally, we conclude the paper with several future directions 
%%in Section~\ref{sec:conc}.  

\section{Problem Statement}\label{sec:prob}

\paragraph{Hypotheticals.} Fix a relational schema $\Sigma$. Our goal is to provision queries on $\Sigma$-instances.  A \emph{hypothetical} w.r.t.  $\Sigma$
is a computable function $h$ that maps every $\Sigma$-instance $I$ to a sub-instance $h(I)\subseteq I$.  Formalisms for specifying hypotheticals are
of course of interest (e.g., apply a selection predicate to each table in $I$) but we do not discuss them here because the results in this paper do
not depend on them.

\paragraph{Scenarios.} We will consider analyses (scenario explorations) that start from a finite set $H$ of hypotheticals.  A \emph{scenario} is a
non-empty set of hypotheticals $S\subseteq H$. The result of applying a scenario $S=\set{h_1,\ldots,h_s}$ to an instance $I$ is defined as a
sub-instance $\app{I}{S}=h_1(I)\cup\cdots\cup h_s(I)$. In other words, under $S$, if any $h \in S$ is said to be turned on (similarly, any $h\in
H\setminus S$ is turned off), each turned on hypothetical $h$ will retain the tuples $h(I)$ from $I$.

\begin{definition}[Provisioning]
\label{def:sketch}
Given a query $Q$, to \emph{provision} $Q$ means to design a pair of algorithms: $(i)$ a
\textbf{compression} algorithm that takes as input an instance $I$ and
a set $H$ of hypotheticals, and outputs a data structure $\Gamma$
called a \textbf{sketch}, and $(ii)$ an \textbf{extraction} algorithm
that for any scenario $S\subseteq H$, outputs $Q(\app{I}{S})$ using
only $\Gamma$ (without access to $I$).
\end{definition}
To be more specific, we assume the compression algorithm takes as
input an instance $I$ and $k$ hypotheticals $h_1, \ldots, h_k$ along
with the sub-instances $h_1(I),\ldots,h_k(I)$ that they define. A
hypothetical will be referred to by an index from $\set{1,\ldots,k}$,
and the extraction algorithm will be given scenarios in the form of
sets of such indices. Hence, we will refer to a scenario $S\subseteq
H$ where $S=\set{h_{i_1},\ldots,h_{i_s}}$ by abusing the notation as
$S=\{i_1,\ldots,i_s\}$. Throughout the paper, we denote by $k$ the
number of hypotheticals (i.e. $k:=\card{H}$), and by $n$ the size of the
input instance (i.e., $n := \card{I}$). 

We call such a pair of compression and extraction algorithms a \emph{provisioning scheme}. 
The compression algorithm runs only once; 
the extraction algorithm runs repeatedly for all the scenarios that 
an analyst wishes to explore.  We refer to the time that the compression
algorithm requires as the \emph{compression time}, and the time that extraction algorithm requires for each scenario as the \emph{extraction time}.

The definition above is not useful by itself for positive results because it allows for trivial space-inefficient solutions.
For example, the definition is satisfied when the sketch $\Gamma$ is defined to be a copy of $I$ itself or, as mentioned earlier, a scenario-indexed collection of all the answers. Obtaining the answer for each
scenario is immediate for either case, but such a sketch can be prohibitively large as the number of tuples in $I$ could be enormous, and the number of scenarios is exponential in $k=\card{H}$.

This discussion leads us to consider complexity bounds on the size of
the sketches. 

\begin{definition}[Compact provisioning]
\label{def:prov}
A query $Q$ can be \emph{compactly} provisioned if there exists a
provisioning scheme for $Q$ that given any input instance $I$ and any
set of hypotheticals $H$, constructs a sketch of size
$\poly(k,\log{n})$ bits, where $k:=\card{H}$ and $n:=\card{I}$.
\end{definition}

We make the following important remark about the restrictions 
made in Definitions~\ref{def:sketch} and~\ref{def:prov}. 
\begin{remark}\label{rem:examine} 
At first glance, the requirement that the input
instance $I$ cannot be examined \emph{at all} during extraction may
seem artificial, and the same might be said about the size of the
sketch depending polynomially on $\log{n}$
rather than a more relaxed
requirement.  However, we show in Theorem~\ref{thm:full-cover}
that our central lower bound result holds
\emph{even if} a portion of size $o(n)$ of the input instance can be examined 
during extraction \emph{after} the scenario is revealed, 
and \emph{even if} the space dependence of the sketch is only restricted 
to be $o(n)$  (instead of depending only polynomially on $\log{n}$). 
These additional properties transfer to \emph{all} our lower bound results
(i.e., Theorems~\ref{thm:csa-lower},~\ref{thm:q-lower},~\ref{thm:reg-lower},~\ref{thm:ucq-dl-lower}) although we choose not to restate them in each of them.
In spite of these additional properties, the positive results we obtain (i.e., Theorems~\ref{thm:count},~\ref{thm:sum},~\ref{thm:regression},~\ref{thm:mix}) all use sketches whose
space requirements depend
polynomially only on $\log{n}$ and do \emph{not} require examining the original database \emph{at all} during the extraction. These calibration results
further justify our design choices for compact provisioning.
\end{remark}

Even though the definition of compact provisioning does not impose any
restriction on either the compression time or the extraction time,
all our positive results in this paper are supported by (efficient) 
polynomial time algorithms. Note that this is \emph{data-scenario
  complexity}: we assume the size of the query (and the schema) to be
a constant but we consider dependence on the size of the instance and
the number of hypotheticals. 
Our negative results (lower bounds on the sketch size), on the other
hand, hold even when the compression and the extraction algorithms are
computationally unbounded.

%%Moreover, we remark that the fact that the original input $I$ is inaccessible 
%%during extraction only strengthens our positive results, but 
%%might be considered unduly restrictive.
%%However, we prove our lower bounds in a more general 
%%setting that allows the extraction algorithm to re-examine a 
%%large (but still $o(|I|)$) fraction of $I$ 
%%\emph{after} being given the scenario $S$.
%%Thus, the lower bounds show that
%%short of revisiting essentially the entire $I$, 
%%(and thus making provisioning moot)
%%the algorithms must use exponential-size sketches.

\paragraph{Exact vs. approximate provisioning.} Definition~\ref{def:prov}
focused on exact answers for the queries. While this is appropriate
for, e.g., relational algebra queries,
as we shall see, for queries that compute numerical answers such as
aggregates and analytics, having the flexibility of answering queries
approximately is essential for any interesting positive result.

\begin{definition}[$\eps$-provisioning]
\label{def:approx}
For any $0<\eps<1$, a query $Q$ can be \emph{$\eps$-provisioned} if
there exists a provisioning scheme for $Q$, whereby for each scenario $S$,
the extraction algorithm outputs a $(1 \pm \eps)$ approximation of
$Q(\app{I}{S})$, where $I$ is the input instance.

We say a query $Q$ can be \emph{compactly} $\eps$-provisioned if $Q$
can be $\eps$-provisioned by a provisioning scheme that, given any
input instance $I$ and any set of hypotheticals $H$, creates a sketch
of size $\poly(k,\log{n},1/\eps)$ bits.
\end{definition}

We emphasize that throughout this paper, we only consider the
approximation guarantees which are \emph{relative} (multiplicative) as
opposed to the weaker notion of additive approximations.  The precise
definition of relative approximation guarantee will be provided for
each query individually.  Moreover, as expected, randomization will be
put to good use in $\eps$-provisioning. We therefore extend the
definition to cover the provisioning schemes that use both
randomization and approximation. 

\begin{definition}\label{def:eps-scheme}
For any $\eps,\delta>0$, an \emph{$(\eps,\delta)$-provisioning scheme}
for a query $Q$ is a provisioning scheme where both the compression and 
extraction algorithms
are allowed to be randomized and the output \emph{for every} scenario 
$S$ is a $(1 \pm \eps)$-approximation of $Q(\app{I}{S})$ with probability 
$1-\delta$.

An $(\eps,\delta)$-provisioning scheme is called \emph{compact} if it
constructs sketches $\Gamma$ of size only\\
$\poly(k,\log{n},1/\eps,\log(1/\delta))$ bits, has
compression time that is $\poly(k,n,1/\eps,\log{(1/\delta)})$, 
and has extraction time that is $\poly(\card{\Gamma})$. 
\end{definition}

In many applications, the size of the database is a very large number,
and hence the exponent in the $\poly(n)$-dependence of the
compression time might become an issue. If the dependence of the
compression time on the input size is essentially linear, i.e.,
$O(n)\cdot
\poly(k,\log{n},1/\eps,\log{(1/\delta)})$ we say that
the scheme is \emph{linear}.
We emphasize that in all our
positive results for queries with numerical answers we give compact
\edSchemes, thus ensuring efficiency in both running time and sketch
size.

%Therefore, we further define \emph{\edScheme,} where the
%dependence of the compression time on $\card{I}$ is \emph{essentially
%  linear}, i.e., 

\paragraph{Complex queries.} 
Our main target consists of practical queries that combine logical,
grouping, and numerical components. In Section~\ref{sec:mix}, we focus
on \emph{complex queries} defined by
a logical (relational algebra or Datalog) query that returns a set
of tuples, 
followed by a group-by operation (on set of grouping attributes)
and further followed by numerical query that is applied to each
sets of tuples resulting from the grouping. This class of queries
already covers many practical examples. 
We observe that the output of such a complex query is a set of $p$
tuples where $p$ is \emph{the number of distinct values taken by the
  grouping attributes}.  Therefore, the size of any provisioning
sketches must also depend on $p$.  We show (in Theorem~\ref{thm:mix}) 
that a sketch for a query that involves
grouping can be obtained as a collection of $p$ sketches.
Hence, if each
of the $p$ sketches is of compact size (as in
Definitions~\ref{def:prov} and~\ref{def:eps-scheme}) and the value $p$
itself is bounded by 
%$\poly(\card{H},\log{\card{I}})$, 
$\poly(k,\log{n})$, then the
overall sketch for the complex query is also of compact size.  Note
that $p$ is typically small for the kind of grouping used in
practical analysis queries (e.g., number of products, number
of departments, number of locations, etc.). Intuitively, an analyst
would have trouble making sense of an output with a large number of
tuples. 

%%Therefore, in the rest of this paper, \emph{we will always
%%assume that $p$ is bounded by} $\poly(\card{H},\log{\card{I}})$.

\paragraph{Notation.}  For any integer $m > 0$, $[m]$ denotes the set
$\set{1, 2, \dots, m}$. The $\tilde{O}(\cdot)$ notation 
suppresses $\log{\log(n)}$, $\log\log(1/\delta)$, $\log(1/\eps)$, and
$\log(k)$ factors. All logarithms are in base two unless stated otherwise.

\newcommand{\bA}{\ensuremath{\mathbf{A}}}
\newcommand{\supp}[1]{\ensuremath{S_{#1}}}
\newcommand{\bB}{\ensuremath{\mathbf{B}}}
\newcommand{\bC}{\ensuremath{\mathbf{C}}}
\newcommand{\Ex}{\ensuremath{\text{E}}}

\newcommand{\bX}{\mathbf{X}}
\newcommand{\bSi}{\pmb{\Sigma}}

\newcommand{\bS}[1]{\mathbf{S}_{#1}}
\newcommand{\bSJ}{\mathbf{S}^{J}}
\newcommand{\bQ}{\mathbf{Q}}
\newcommand{\bD}{\mathbf{D}}
\newcommand{\bJ}{\mathbf{J}}
\newcommand{\bT}{\pmb{\theta}}
\newcommand{\bI}{\mathbf{I}}

\newcommand{\GS}{\Gamma_{\Sigma}}
\newcommand{\Range}[1]{\bSi_D(#1)}
\newcommand{\exam}{\texttt{exam}}
\newcommand{\event}{\mathcal{E}}

\newcommand{\Sall}{\mathcal{X}}
\newcommand{\Sgood}{\Sall_g}
\newcommand{\dist}{\mathcal{D}}
\newcommand{\distu}{\ensuremath{\mu}}
\newcommand{\distv}{\ensuremath{\nu}}
\newcommand{\KLD}[2]{\ensuremath{\text{D}(#1\parallel#2)}}
\newcommand{\bracket}[1]{\left[#1\right]}

\section{A ``Hard'' Problem for Provisioning} \label{sec:full-cover}

To establish our lower bounds in this paper, we introduce a ``hard'' problem called \fullcover. 
Though not defined in the exact formalism of provisioning, the \fullcover problem can be solved by many provisioning
schemes using proper ``reductions'' and hence a lower bound for the \fullcover problem can be used to establish similar lower bounds for provisioning various queries.

We start by describing the \fullcover problem in details and then present our lower bound. 
In Appendix~\ref{app:info}, we survey some simple tools from information theory that we need in our lower bound proof. 

\subsection{The Coverage Problem}

Informally speaking, in the \fullcover problem, we are given a collection of $k$ subsets of a universe $[n]$ and the goal is to ``compress'' this collection in order to answer to the questions in which indices
 of some subsets in the collection are provided and we need to figure out whether these subsets cover the universe $[n]$ or not. We are interested in compressing schemes for this problem that when answering each question, in addition to the already computed summary of the collection, also have a limited access to the original instance (see Remark~\ref{rem:examine} for motivation of this modification). The \fullcover problem is defined formally as follows. 

\begin{problem}[\fullcover]\label{prob:full-cover}
  Suppose we are given a collection $\mathcal{S} = \set{S_1, S_2, \ldots S_k}$ of
  the subsets of $[n]$. The goal in the \fullcover problem is to find a \emph{compressing
  scheme} for $\mathcal{S}$, defined formally as a triple of algorithms:
\begin{itemize}
	\item A \textbf{compression algorithm} which given the collection $\mathcal{S}$ creates a \emph{data structure} $D$.
	\item An \textbf{examination algorithm} which given a subset of $[k]$, a \emph{question}, $Q = \set{i_1,\ldots, i_s}$ and the data structure $D$, computes a set $J \subseteq [n]$ of indices and
	lookup for each $j \in J$ and each $S_i$ ($i \in [k]$), whether $j \in S_i$ or not. The output of the examination algorithm is a tuple $S^J:=(S^J_1,\ldots,S^J_k)$, where $S^J_i = S_i \cap J$.
	\item An \textbf{extraction algorithm} which given a question $\set{i_1,\ldots, i_s}$, the data structure $D$, and the tuple $S^J$, outputs ``Yes'', if $S_{i_1} \union \ldots \union S_{i_s} = [n]$ and ``No'' otherwise.
\end{itemize}
\end{problem}
  
We refer to the size of $D$, denoted by $\card{D}$, as the storage requirement of the compression scheme and to the size of $J$, denoted by $\card{J}$, as the examination requirement of the scheme.
The above algorithms can all be randomized; in that case, we require that for \emph{each} question $Q$, the final answer (of the extraction algorithm) to be correct with a probability at least $0.99$.
Note that this choice of constant is arbitrary and is used only to simplify the analysis; indeed, one can always amplify the probability of success by repeating the scheme constant number of times and return
the majority answer. 

%%Note that allowing access to the original input in Lemma~\ref{lem:full-cover} makes the lower bound very robust. 
%%However, due to this property, the lower bound does not seem to follow from standard communication complexity lower bounds and hence we use an information-theoretic approach
%%to prove this lemma directly, which may be of independent interest. As the proof of this lemma requires quite a detour from the general theme of this paper, we defer the proof 
%%to Section~\ref{sec:full-cover}. We note that since our other lower bounds are typically proven using a reduction from the \fullcover problem, the properties in Lemma~\ref{lem:full-cover} also hold
%%for them and we do not mention this explicitly. 

While \fullcover is not stated in the exact formalism of provisioning, the analogy between this problem and provisioning schemes should be clear. In particular, 
for our lower bound proofs for provisioning schemes, we can alter the Definition~\ref{def:sketch} to add an examination algorithm and allow a similar access to the original database
to the provisioning scheme.

\subsection{The Lower Bound}

We prove the following lower bound on storage and examination requirement of any compressing scheme for the \fullcover problem. 

\begin{theorem} \label{thm:full-cover}
	Any compressing scheme for the \fullcover problem that answers each question correctly with probability at least $0.99$, either has storage requirement or examination requirement of 
	$\min(2^{\Omega{(k)}},\Omega(n))$ bits. 
\end{theorem}

Allowing access to the original input in Theorem~\ref{thm:full-cover} makes the lower bound very robust. 
However, due to this property, the lower bound does not seem to follow from standard communication complexity lower bounds and hence we use an information-theoretic approach
to prove this theorem directly, which may be of independent interest. We note that since our other lower bounds are typically proven using a reduction from the \fullcover problem, the properties
in Theorem~\ref{thm:full-cover} (i.e., allowing randomization and $o(n)$ access to the database \emph{after} being given the scenario) also hold
for them and we do not mention this explicitly.

In order to prove this lower bound, by Yao's minimax principle~\cite{Yao77} (see also~\cite{MotwaniR97}), it suffices to define a distribution of instances of the problem and show that any \emph{deterministic}
algorithm that is correct on the input chosen according to this distribution with probability $0.99$, has to have either large storage requirement or large examination requirement. We define the following distribution $\dist$  (for simplicity assume $k$ is even).

\cTextbox{Hard distribution $\dist$ for the \fullcover problem.}{
\begin{enumerate}
\item Let $n = {k \choose k/2}$; pick a string $x \in \set{0,1}^n$ uniformly at random.
\item Let $\FC$ be the family of all subsets of $k$ with size $k/2$. Pick uniformly at random a \emph{bijection} $\sigma: \FC \mapsto [n]$.
\item \textbf{Embedding.} Starting from $S_i = [n]$ for all $i \in [k]$, for any $j \in [n]$, if $x_j = 0$, remove $j$ from all sets $S_i \in \sigma^{-1}(j)$.
\item \label{line:question} \textbf{Question.} Pick the question $Q$ to be a uniformly at random member of $\FC$.
\end{enumerate}
}

The following claim is immediate.
\begin{claim}\label{clm:op}
	For any question $Q = \set{S_{i_1},\ldots,S_{i_{k/2}}} \in \FC$,  $S_{i_1} \union \ldots \union S_{i_{k/2}} = [n]$ iff $x_{\sigma(Q) }= 1$.
\end{claim}

Fix any deterministic compressing scheme for the distribution $\dist$. We define $\bD$, $\bQ$,  and $\bJ$ as the random variables corresponding to, respectively, the data structure $D$, the question $Q$, and the
examination indices $J$ in this compressing scheme.
Moreover $\bX$ is a random variable for the input string $x$, $\bSi$ is for the bijection $\sigma$, and $\bS{i}$ ($i \in [k]$) is for the subset $S_i$.
We use $\bSJ := (\bSJ_1,\ldots,\bSJ_{k})$ for the output of the examination algorithm.

\paragraph{Overview of the proof.} The intuition behind the proof is as follows. By Claim~\ref{clm:op}, in order to give a correct answer
to the question $Q$, the scheme has to determine the value of $x_{\sigma(Q)}$ correctly.
Suppose first that the scheme only consists of compression and extraction algorithms, i.e., without the examination algorithm. In this case, even if we give the bijection $\sigma$ to the extraction algorithm,
the extraction algorithm has to effectively recover the value of $x_{\sigma(Q)}$ from $D$, where $\sigma(Q)$ is chosen uniformly at random from $[n]$. In this case, simple information-theoretic facts imply 
that $D$ has to be of size $\Omega(n)$.

Now consider the other case where we remove the compression algorithm. In this case, even if we give the string $x$ to the examination algorithm
and assume that upon examining an entry $j$ in the input, it can determine whether $\sigma(Q) = j$ or not, for the extraction algorithm to be able to find the value of $x_{\sigma(Q)}$ correctly,
it better be the case that $\sigma(Q) \in J$. In other words, the set of indices computed by the examination algorithm should contain the target index $\sigma(Q)$. However, for any fixed $J$ of size $o(n)$, the
probability that $\sigma(Q) \in J$ is $o(1)$, and hence the extraction algorithm cannot recover the correct answer with high probability.

To prove Theorem~\ref{thm:full-cover}, we have to consider schemes that consist of both compression and examination algorithms and a priori it is not clear that how much the interleaved information obtained from
both these algorithms can help the extraction algorithm to compute the final answer, especially considering the fact that having a compression algorithm also helps examination algorithm in finding the ``correct'' index.
However, using a careful analysis we separate the information given from these two algorithms to the extraction algorithm and argue that
at least one of the examination or compression requirements of any scheme has to be of size $\Omega(n)$. 

\begin{proof}[Proof of Theorem~\ref{thm:full-cover}]
	Suppose $\bT$ is the random variable which is $1$ if the correct answer is Yes, and is zero otherwise and $\bI$ is the random variable that denotes the index of the string $x$ which determines the correct answer, i.e., $\bI = \bSi(\bQ)$
	(by Claim~\ref{clm:op}). Let $\delta$ $(\le 0.01)$ be the probability of failure; we have,
	\begin{align*}
		H_2(\delta) &\geq H(\bT \mid \bD, \bQ, \bSJ) \tag{Fano's inequality, Claim~\ref{clm:it-facts}-(\ref{part:fano})}  \\
		&\geq H(\bT \mid \bD, \bQ, \bSJ, \bI) \tag{Conditioning reduces entropy, Claim~\ref{clm:it-facts}-(\ref{part:cond-reduce}) } \\
		&= H(\bT \mid \bD,\bQ,\bI) - I(\bT;\bSJ \mid \bD,\bQ,\bI) 
	\end{align*}
	
	We bound each term in the above equation in the following two claims separately.
	\begin{claim}\label{clm:low-end}
		Suppose $\card{\bD} = o(n)$; then $H(\bT \mid \bD,\bQ,\bI) = 1-o(1)$.
	\end{claim}
	\begin{proof}
	\begin{align*}
		H(\bT \mid \bD,\bQ,\bI) &\geq H(\bT \mid \bD,\bQ,\bI,\bSi) \tag{Conditioning reduces entropy, Claim~\ref{clm:it-facts}-(\ref{part:cond-reduce}) } \\
			&= H(\bT \mid \bD,\bI,\bSi)  \tag{$\bQ$ is uniquely determined by $\bSi$ and $\bI$} \\
			&= \sum_{i=1}^{n} \Pr(\bI = i) \cdot H(\bT \mid \bD,\bI = i,\bSi) \\
			&= \sum_{i=1}^{n} \frac{1}{n} \cdot H(\bX_i \mid \bD, \bI = i, \bSi) \tag{$\bI$ is uniform on $[n]$ and $\bT = \bX_i$} \\
	\end{align*}
		Now, notice that $\bSi$ is independent of the event $\bI = i$, and moreover conditioned on $\bSi$, $\bD$ is a function of $\bX$ alone and hence is independent of $\bI = i$. Additionally, $\bX_i$ is chosen independent
		of the value of $\bI = i$; hence $\bX_i$ is also independent of $\bI = i$. Consequently, we can drop conditioning on $\bI = i$ and have,
	\begin{align*}	
			H(\bT \mid \bD,\bQ,\bI) &\geq \sum_{i=1}^{n} \frac{1}{n} \cdot H(\bX_i \mid \bD, \bSi)  \\
			&\geq \frac{H(\bX \mid \bD,\bSi)}{n} \tag{By subadditivity of entropy, Claim~\ref{clm:it-facts}-(\ref{part:sub-additivity})} \\
			&= \frac{H(\bX \mid \bSi) - I(\bX,\bD \mid \bSi)}{n}  \\
			&\geq \frac{H(\bX) - H(\bD)}{n}  \tag{$\bX \perp \bSi$ and $I(\bX,\bD \mid \bSi) \leq H(\bD \mid \bSi) \leq H(\bD)$} \\
			&\geq 1-\frac{\card{\bD}}{n}  = 1-o(1)
	\end{align*}
	where in the last inequality we use the fact that $\bX$ is uniformly chosen from a domain of size $2^n$ and hence $H(\bX) = n$ and $H(\bD) \leq \card{\bD}$ (both by Claim~\ref{clm:it-facts}-(\ref{part:uniform})).
	\end{proof}

	\begin{claim}\label{clm:low-info}
		Suppose $\card{\bD} + \card{\bJ} = o(n)$ ; then $ I(\bT;\bSJ \mid \bD,\bQ,\bI) \leq 0.9$.
	\end{claim}
	\begin{proof}
		We first define some notation. For a fixed $D$ and $Q$, we
		use $J = \exam(D,Q)$ to denote the unique set of examined indices (recall that the compressing scheme is deterministic over the distribution $\dist$). 
		For a triple $T:=(D,Q,i)$ as an assignment to $(\bD,\bQ,\bI)$, we say that the tuple is \emph{good} if $i \notin J$, where $J = \exam(D,Q)$. We use the set $\Sall$ to denote the set of
		all valid tuples $T$ and the set $\Sgood$ to denote the set of all good tuples. Using this notation, 
		\begin{align*}
			I(\bT;\bSJ \mid \bD,\bQ,\bI) &= \sum_{T \in \Sall} \Pr((\bD,\bQ,\bI) = T) \cdot I(\bT;\bSJ \mid (\bD,\bQ,\bI)=T)
		\end{align*}
		We decompose the above summands into two parts; one over $\Sgood$ and one over $\Sall \setminus \Sgood$. We first argue that for a tuple $T \in \Sgood$,
		$I(\bT;\bSJ \mid (\bD,\bQ,\bI)=T) = 0$. This is because conditioned on $(\bD,\bQ,\bI) = (D,Q,i)$, by Claim~\ref{clm:op}, $\bT = \bX_i$ and $\bX_i$ is independent of $\bX \setminus \bX_i$ and $\bSi$, and hence
		$\bX_i$ is independent of $\bSJ$ as well whenever $i \notin J$. 
		For a tuple $T \notin \Sgood$, we simply use the upper bound $I(\bT;\bSJ \mid (\bD,\bQ,\bI)=T) \leq H(\bT \mid (\bD,\bQ,\bI)=T ) \leq \card{\bT}$, where $\card{\bT} = 1$.
		Consequently,
		\begin{align*}
			I(\bT;\bSJ \mid \bD,\bQ,\bI) &= \sum_{T \notin \Sgood} \Pr((\bD,\bQ,\bI) = T) \cdot I(\bT;\bSJ \mid (\bD,\bQ,\bI)=T) \\
			&\leq \sum_{T \notin \Sgood} \Pr((\bD,\bQ,\bI) = T)  \cdot \card{\bT} \\
			&=  \Pr((\bD,\bQ,\bI) \notin \Sgood) \tag{$\card{\bT} =1$}
		\end{align*}
		 We now show that $\Pr((\bD,\bQ,\bI) \notin \Sgood) \leq 0.9$. Let $\distu$ be the distribution of $\bI$ (i.e., a uniform distribution on $[n]$) and for any $D$ and $Q$, $\distv_{D,Q}$ be the distribution
		  of $\bI \mid \bD = D,\bQ = Q$. Fix any set $J$ of size $o(n)$; 
		 it is clear that under $\distu$, $\Pr_\distu(\bI \in J) = o(1)$. We show that the total variation distance of $\distu$ and $\distv_{D,Q}$ (for ``typical'' choices of $D$ and $Q$)  is bounded away from $1$ and
		 hence $\Pr_{\distv_{D,Q}}(\bI \in J) < 1$ as well (using Claim~\ref{clm:it-TV}). To achieve the bound on the total variation distance, we instead bound the KL-divergence of $\distv_{D,Q}$ from the
		 uniform distribution $\distu$ (in expectation over the choice of $D$ and $Q$) and then use Pinsker's inequality (Claim~\ref{clm:it-KL-TV}) to bound the total variation distance between these distributions. We have,
		 \begin{align*}
		 	\Ex_{D,Q}\bracket{\KLD{\distv_{D,Q}}{\distu}} &= I(I;D,Q) = H(\bI) - H(\bI \mid \bD,\bQ) \tag{by Claim~\ref{clm:it-KL}} \\
			&=  \log{n} - \sum_{Q \in \FC} \Pr(\bQ = Q) \cdot H(\bI \mid \bD,\bQ = Q) \\
			&= \log{n} - \sum_{Q \in \FC} \frac{1}{n} \cdot H(\bSi(Q) \mid \bD) \tag{$I = \bSi(Q)$ and $\bSi(Q) \perp \bQ = Q, \bD \perp \bQ = Q$}\\
			&\leq \log{n} - \frac{H(\bSi \mid \bD)}{n} \tag{By subadditivity of entropy, Claim~\ref{clm:it-facts}-(\ref{part:sub-additivity})} \\
			&\leq \log{n} - \frac{\card{\bSi} - \card{\bD}}{n} \tag{$H(\bSi \mid \bD) \geq \card{\bSi} - \card{\bD}$, Claim~\ref{clm:it-facts}-(\ref{part:uniform})}\\ 
			&=  \log{n} - \frac{\log{(n!)} - o(n)}{n} \tag{$\card{\bSi} = \log{n!}$, $\card{\bD} = o(n)$}\\
			&\leq \log{n} - \frac{n\log{n} - n\log{e} + O(\log{n}) - o(n)}{n} \tag{by Stirling approximation of $n!$} \\
			&= \log{e} + o(1)
		 \end{align*}
		 We now have, 
		 \begin{align*}
		 	\Ex_{D,Q}\bracket{\card{\distv_{D,Q} - {\distu}}} &\leq \Ex_{D,Q}\bracket{\sqrt{\frac{1}{2}\cdot\KLD{\distv_{D,Q}}{\distu}}} \tag{Pinsker's inequality, Claim~\ref{clm:it-KL-TV}} \\
			&\leq \sqrt{\frac{1}{2} \cdot \Ex_{D,Q}\bracket{\KLD{\distv_{D,Q}}{\distu}}} \tag{Convexity of $\sqrt{\cdot}$} \\
			&\leq \sqrt{\frac{\log{e}}{2} + o(1)} < 0.85
		 \end{align*}
		Fix any pair $(D,Q)$, by Claim~\ref{clm:it-TV}, 
		\begin{align*}
			\Pr_{\distv_{D,Q}}(i \in \exam(D,Q)) \leq \Pr_{\distu}(i \in \exam(D,Q)) + \card{\distv_{D,Q} - \distu}
		\end{align*}
		By taking expectation over $(D,Q)$,
		\begin{align*}
			\Ex_{D,Q}\bracket{\Pr_{\distv_{D,Q}}(i \in \exam(D,Q))} &\leq \Ex_{D,Q}\bracket{\Pr_{\distu}(i \in \exam(D,Q)) } + \Ex_{D,Q}\bracket{\card{\distv_{D,Q} - \distu}} \\
			&\leq o(1) + 0.85 \leq 0.9
		\end{align*}
		Noting that the LHS in the above equation is equal $\Pr((\bD,\bQ,\bI) \notin \Sgood)$ finalizes the proof. 
	\end{proof}
	
	By plugging in the values from the above two claims, we have
	\[
			H_2(\delta) \geq H(\bT \mid \bD,\bQ,\bI) - I(\bT;\bSJ \mid \bD,\bQ,\bI) \geq 1-o(1) - 0.9 = 0.1 - o(1)
	\]
	which is in contradiction with the fixed value of $H_2(\delta) \sim 0.08$ (for $\delta = 0.01$). Hence, $\card{\bD} + \card{\bJ} = \Omega(n)$, implying that the
	storage requirement or examination requirement of the scheme has to be of size $\Omega(n)$. The bound of $2^{\Omega(k)}$ also follows from the fact that $n = 2^{\Omega(k)}$ in this construction.
	
\end{proof}

Notice that in the proof of Theorem~\ref{thm:full-cover}, the constructed
instances have a simple property (the union in each question is either $[n]$ or $[n]\setminus\set{i}$ for
some $i$).  We extract this useful property and provide the following
stronger version of Theorem~\ref{thm:full-cover} as a corollary (to its proof). This stronger version is used to prove a lower bound for the 
\avgQuery query in Theorem~\ref{thm:csa-lower}.

\begin{corollary}\label{lem:strict-full-cover}
  Suppose the given sets $S_1,\ldots,S_{k}$ in the $\fullcover$ problem
  are promised to have the property that for any set $\set{i_1,\ldots,i_{\floor{k/2}}}
  \subseteq [k]$, $S_{i_1} \union \ldots \union S_{i_\floor{k/2}}$ is
  either $[n]$ or $[n] \setminus \set{j}$ for some $j \in [n]$. The lower
  bound of Theorem~\ref{thm:full-cover} still holds for this promised version of
  the $\fullcover$ problem.
\end{corollary}

\section{Numerical Queries}\label{sec:num-qry}
In this section, we study provisioning of numerical queries, i.e.,
queries that output some (rational) number(s) given a set of tuples. In particular, we investigate aggregation queries
including  count, sum, average, and quantiles (therefore min, max,
median, rank, and percentile), and as a first step towards
provisioning database-supported machine learning, linear $(\ell_2)$ regression.
We assume that
the relevant attribute values are rational numbers of the form $a/b$
where both $a,b$ are integers in range $[-W, W]$ for some $W > 0$.
\subsection{The Count, Sum, and Average Queries} \label{sec:csa}
In this section, we study provisioning of the \countQuery, \sumQuery, and \avgQuery queries, formally
defined as follows. The answer to the \countQuery query is the number of tuples in the input
instance. For the other two queries, we assume a relational
schema with a binary relation containing two attributes:
an \emph{identifier (key)} and a \emph{weight}.
We say that a tuple $x$ is smaller than the tuple $y$, if the weight of
$x$ is smaller than the weight of $y$. Given an instance $I$, the
answer to the \sumQuery query (resp. the \avgQuery query) is the \emph{total weights} of the
tuples (resp. the \emph{average weight} of the tuples) in $I$.

We first show that none of the \countQuery, \sumQuery, \avgQuery queries can be provisioned both compactly and exactly, which motivates the $\eps$-provisioning approach, and then briefly describe how to build a compact \edScheme for each of them. Our lower bound results for \countQuery, \sumQuery, and \avgQuery queries are summarized in the following theorem.

\begin{theorem}\label{thm:csa-lower}
  Exact provisioning of any of the \countQuery, \sumQuery, or \avgQuery queries requires sketches of size
  $\min(2^{\Omega(k)},\Omega(n))$ bits.
\end{theorem}
\begin{proof}
We provide a proof for each of the queries separately. 

\textbf{Count query.} Given $\set{S_1,\ldots,S_k}$, where each $S_i$ is a subset of $[n]$, we solve \fullcover using a provisioning scheme for the \countQuery query. Define an instance $I$ of a
relational schema with a unary relation $A$, where $I = \set{A(x)}_{x
  \in [n]}$. Define a set $H$ of $k$ hypotheticals, where for any $i
\in [k]$, $h_i(I) = \set{A(x)}_{x \in S_i}$. For any scenario
$S=\set{i_1,\ldots,i_s}$, the count of $\app{I}{S}$ is $n$ iff
$S_{i_1} \union \ldots \union S_{i_s} = [n]$. Hence, any provisioning
scheme for the \countQuery query solves the $\fullcover$ problem and
the lower bound follows from Theorem~\ref{thm:full-cover}.

\textbf{Sum query.} The lower bound of the \sumQuery follows immediately from the one for \countQuery by setting all weights to be $1$.

\textbf{Average query.} For simplicity, in the following construction we will omit the id attribute of the tuples as the weights of the tuples
are distinct and can be used to identify each tuple. 

Given $\mathcal{S} = \set{S_1, S_2, \ldots S_k}$, where each $S_i$
is a subset of $[n]$, with the promise that for any set
$\set{i_1,\ldots,i_{\floor{k/2}}} \subseteq [k]$, $S_{i_1} \union
\ldots \union S_{i_\floor{k/2}}$ is either $[n]$ or $[n] \setminus
\set{j}$ for some $j \in [n]$, we want to solve this restricted
$\fullcover$ problem. By Corollary~\ref{lem:strict-full-cover}, the
restricted $\fullcover$ problem also needs a data structure of size
$\min(\Omega(n),2^{\Omega(k)})$ bits.

Define an instance $I$ of the relational schema with a unary relation
$A$, where $I = \set{A(x)}_{x \in [n]}$. Define a set $H$ of $k$
hypotheticals, where for any $i \in [k]$, $h_i(I) = \set{A(x)}_{x \in
  S_i}$. For any scenario $S = \set{i_1,\ldots,i_{\floor{k/2}}}$, the
average weight of $\app{I}{S}$ is ${n+1 \over 2}$
(resp. $\frac{n(n+1)/2 - j}{n-1}$ for some $j \in [n]$) if $S_{i_1}
\union \ldots \union S_{i_{\floor{k/2}}}$ is equal to (resp. not equal
to) $[n]$. The two values are equal iff $j = (n+1)/2$, which, if we
assume $n$ is even, could never happen. Therefore knowing the average
value is enough to solve the restricted $\fullcover$ problem, and the
lower bound follows.
\end{proof}

We further point out that if the weights can be both positive and negative, the \sumQuery (and \avgQuery) cannot be compactly provisioned even approximately, and hence we will
focus on $\eps$-provisioning for \emph{positive} weights.

\begin{theorem}\label{thm:neg-sum}
  Provisioning of the \sumQuery (and \avgQuery) query approximately (up to any
  multiplicative factor) over the input instances with both positive
  and negative weights requires sketches of size
  $\min(2^{\Omega(k)},\Omega(n))$ bits.
\end{theorem}
 \begin{proof}
  We use a reduction from the $\fullcover$ problem. Suppose we are
  given a collection of sets $\mathcal{S} = \set{S_1, S_2, \ldots S_k}$, where each
  $S_i$ is a subset of $[n]$ in the $\fullcover$ problem. Consider the
  relational schema $\Sigma = \set{A}$ where $A$ is binary and the
  second attribute of $A$ is the weight. Let $I = \set{A(1,1), A(2,1),
    \ldots, A(n,1), A(n+1,-n)}$.  We define $k$ hypotheticals: for any
  $i \in [k]$, $h_i(I) = \set{A(j,1)~|~j\in S_i} \union
  \set{A(n+1,-n)}$.  For any set $\set{i_1,\ldots,i_s} \subseteq [k]$,
  let $S$ be the scenario consisting of $h_{i_1},\ldots,h_{i_s}$; then
  the total weight of $\app{I}{S}$ is zero iff $S_{i_1} \union \ldots
  \union S_{i_s} = [n]$.  Therefore, any multiplicative approximation
  to the \sumQuery query would distinguish between the zero and
  non-zero cases, and hence solves the $\fullcover$ problem.  The
  lower bound on the size of the sketch now follows from
  Theorem~\ref{thm:full-cover}.
\end{proof}

We conclude this section by explaining the $\eps$-provisioning schemes for the
\countQuery, \sumQuery, and \avgQuery queries.  Formally, 
\begin{theorem}[$\eps$-provisioning \countQuery]\label{thm:count}
For any $\eps,\delta > 0$, there exists a compact \edScheme for the
\countQuery query that creates a sketch of size
$\tilde{O}(\eps^{-2}k(k+\log(1/\delta)))$ bits.
\end{theorem}

\begin{theorem}[$\eps$-provisioning \sumQuery \& \avgQuery]\label{thm:sum}
For instances with positive weights, for any $\eps,\delta > 0$, there
exists compact \edSchemes for the \sumQuery and \avgQuery queries,
with a sketch of size
$\tilde{O}(\eps^{-2}k^2\log(n)\log(1/\delta) + k\log\log{W})$ bits.
\end{theorem}

We remark that the results in Theorems~\ref{thm:count} and~\ref{thm:sum} are mostly direct application of known techniques and are presented 
here only for completeness.

The \countQuery query can be provisioned by using \emph{linear sketches} for estimating the $\ell_0$-norm (see, e.g.,~\cite{kane2010optimal}) as follows. Consider each hypothetical $h_i(I)$ as an $n$-dimensional boolean vector $\bx_i$, where the $j$-th entry is $1$ iff the $j$-th tuple in $I$ belongs to $h_i(I)$.  For each $\bx_i$, create a linear sketch (using $\Ot(\eps^{-2}\log{n})$ bits of space) that estimates the $\ell_0$-norm~\cite{kane2010optimal}. Given any scenario $S$, combine (i.e., add together) the linear sketches of the hypotheticals in $S$ and use the combined sketch to estimate the $\ell_0$-norm (which is equal to the answer of \countQuery).

\begin{remark}
Note that we can directly use linear sketching for provisioning the \countQuery query since counting the duplicates once (as done by union) or multiple times (as done by addition) does not change the answer. However, this is \emph{not} the case for other queries of interest like \quantileQuery and \regQuery and hence linear sketching is not directly applicable for them.
\end{remark}

Here, we also describe a self-contained approach for $\eps$-provisioning the \countQuery query with a slightly better
dependence on the parameter $n$ ($\log\log{n}$ instead of $\log{n}$).

We use the following fact developed by~\cite{bar2002counting} in the
streaming model of computation to design our scheme.  For a bit string
$s \in \set{0,1}^+$, denote by \trail{s} the number of trailing $0$'s
in $s$.  Given a list of integers $A = (a_1,\ldots,a_n)$ from the
universe $[m]$, a function $g:$ $[m]\rightarrow [m]$, and an integer
$t > 0$, the \thtrail{t}{g} of $A$ is defined as the list of the $t$
smallest $\trail{g(a_i)}$ values (use binary expression of $g(a_i)$),
where the duplicate elements in $A$ are counted only once.

\begin{cLemma}[\cite{bar2002counting}]\label{lem:stream-count}
Given a list $A=(a_1,\ldots,a_n)$, $a_i \in [m]$ with $F_0$ distinct
elements, pick a \emph{random pairwise independent hash function}
$\hash:$ $[m] \rightarrow [m]$, and let $t = \lrceil{256\eps^{-2}}$.
If $r$ is the largest value in the \thtrail{t}{\hash} of $A$ and $F_0
\geq t$, then with probability at least $1/2$, $t \cdot 2^r$ is a
$(1\pm \eps)$ approximation of $F_0$.
\end{cLemma}

We now define our \edScheme for the \countQuery query.

\cTextbox{Compression algorithm for the \countQuery query.}{ Given an input instance $I$, a set $H$ of
hypotheticals, and an $\eps > 0$:
\begin{cEnumerate}
    \item Assign each tuple in $I$ with a unique number (an
      identifier) from the set $[n]$.
    \item Let $t = \lrceil{256\eps^{-2}}$. Pick $\lrceil{k +
      \log{(1/\delta)}}$ random pairwise independent hash functions
      $\set{\hash_j: [n] \rightarrow [n]}_{i = 1}^{\lrceil{k +
        \log{(1/\delta)}}}$.  For each hash function $\hash_j$, create
      a sub-sketch as follows.
        \begin{enumerate}
          \item Compute the \thtrail{t}{\hash_j} over the identifiers
            of the tuples in each $h_i(I)$.
          \item Assign a new identifier to any tuple that accounts for
            at least one value in the \thtrail{t}{\hash_j} of any
            $h_i$, called the \emph{concise identifier}.
          \item For each hypothetical $h_i$, \record each
            value in the \thtrail{t}{\hash_j} along with the concise
            identifier of the tuple that accounts for it.
        \end{enumerate}
\end{cEnumerate}
}

\cTextbox{Extraction algorithm for the \countQuery query.}{ 
Given a scenario $S$, for each hash function
$\hash_j$, we use the concise identifiers to compute the union of the
\thtrail{t}{\hash_j} of the hypotheticals that are turned on by
$S$. Let $r$ be the $t$-th smallest value in this union, and compute $t
\cdot 2^r$. Output the median of these $t \cdot 2^r$ values among all
the hash functions. 
}

We call a sketch created by the above compression algorithm a
\countSketch.  For each hash function $\hash_j$ and each $h_i(I)$, we
record concise identifiers and the number of trailing $0$'s
($O(\log{\log{n}})$ bits each) of at most $t$ tuples. Since at most
$t\cdot k$ tuples will be assigned with a concise identifier,
$O(\log{(t\cdot k)})$ bits suffice for describing each concise
identifier. Hence the total size of a \countSketch is
$\tilde{O}(\eps^{-2}k\cdot (k+\log(1/\delta)))$ bits. We now prove the
correctness of this scheme.

\begin{proof}[Proof of Theorem~\ref{thm:count}]
Fix a scenario $S$; for any picked hash function $\hash_i$, since the
$t$-th smallest value of the union of the recorded
\thtrail{t}{\hash_i}, $r$, is equal to the $t$-th smallest value in
the \thtrail{t}{\hash_i} of $\app{I}{S}$. Hence, by
Lemma~\ref{lem:stream-count}, with probability at least $1/2$, $t
\cdot 2^r$ is a $(1 \pm \eps)$ approximation of
$\card{\app{I}{S}}$. By taking the median over $\lrceil{k +
  \log{(1/\delta)}}$ hash functions, the probability of failure is at
most $\delta / 2^k$. If we take union bound over all $2^k$ scenarios,
with probability at least $1 - \delta$, all scenarios could be
answered with a $(1 \pm \eps)$ approximation.
\end{proof}

We now state the scheme for provisioning the \sumQuery query; the
schemes for the \sumQuery and the \countQuery queries together can
directly provision the \avgQuery query, which finalizes the proof of Theorem~\ref{thm:sum}. We use and extend our
\countSketch to $\eps$-provision the \sumQuery query.

\cTextbox{Compression algorithm for the \sumQuery query.} 
{
Given an instance $I$, a set $H$ of
hypotheticals, and two parameters $\eps,\delta > 0$, let $\eps' =
\eps/4$, $t = \lrceil{\log_{1+\eps'}{(n/\eps')}}$, and $\delta' =
      {\delta \over k(t+1)}$.
\begin{cEnumerate}
  \item Let $p = {\lrceil{\log_{(1+\eps')}{W}}}$ and for any $l \in [p]$, let
  $\bar{w}_l = (1+\eps')^{l}$.  For each $l \in [p]$,
    define a set of $k$ new hypotheticals $H_l = \set{h_{l,1},h_{l,2},\ldots,h_{l,k}}$,
    where $h_{l,i}(I) \subseteq h_{i}(I)$ and contains the tuples whose weights are in
    the interval $[\bar{w}_l, \bar{w}_{l+1})$.
  \item \label{line:sum-prune} For each hypothetical $h_i$, let $w$ be
    the largest weight of the tuples in $h_i(I)$, and find the index
    $\gamma$ such that $\bar{w}_\gamma \le w < \bar{w}_{\gamma+1}$.
    \textbf{Record $\gamma$}, and discard all the tuples in $h_i(I)$ with
    weight less than $\bar{w}_{\gamma-t}$.\footnote{In case $\gamma <
      t$, no tuple needs to be discarded.} Consequently, all the remaining
    tuples of $h_i(I)$ lie in the $(t+1)$ intervals
    $\set{[\bar{w}_l, \bar{w}_{l+1})}_{l = \gamma-t}^{\gamma}$. We
      refer to this step as the \emph{pruning step}.
  \item For each $l$, denote by $\hat{H}_l$ the resulting set of
    hypotheticals after discarding the above small weight tuples from
    $H_l$ (some hypotheticals might become empty).  For each of the
    $\hat{H}_l$ that contains at least one non-empty hypothetical, run
    the compression algorithm that creates a \countSketch for $I$ and
    $\hat{H}_l$, with parameters $\eps'$ and
    $\delta'$. \textbf{Record} each created \countSketch.
\end{cEnumerate}
}

\cTextbox{Extraction algorithm for the \sumQuery query.}{
Given a scenario $S$, for any interval
$[\bar{w}_l,\bar{w}_{l+1})$ with a recorded \countSketch, compute the
  estimate of the number of tuples in the interval, denoted by
  $\tilde{n}_l$.  Output the summation of the values $(\bar{w}_{l+1}
  \cdot \tilde{n}_{l})$, for $l$ ranges over all the intervals
  $[\bar{w}_l, \bar{w}_{l+1})$ with a recorded \countSketch.
}
We call a sketch created by the above provisioning scheme a
\sumSketch.  Since for every hypothetical we only record the $(t+1)$
non-empty intervals, by an amortized analysis, the size of the sketch
is $\tilde{O}(\eps^{-2}k^2\log(n)\log(1/\delta) + k\log\log{W})$ bits.
We now prove the correctness of this scheme.

\begin{proof}[Proof of Theorem~\ref{thm:sum}]
For now assume that we do not perform the pruning step (line
(\ref{line:sum-prune}) of the compression phase). For each interval
$[\bar{w}_l,\bar{w}_{l+1})$ among the $\lrceil{\log_{1+\eps'}{W}}$
  intervals, the \countSketch outputs a $(1 \pm \eps')$ approximation
  of the total number of tuples whose weight lies in the
  interval. Each tuple in this interval will be counted as if it has
  weight $\bar{w}_{l+1}$, which is a $(1 + \eps')$ approximation of
  the original tuple weight. Therefore, we can output a $(1 \pm
  \eps')^2$ approximation of the correct sum.

Now consider the original \sumSketch with the pruning step. We need to
show that the total weight of the discarded tuples is negligible.
For each hypothetical $h_i$, we discard the tuples whose weights are
less than $\bar{w}_{\gamma-t}$, while the largest weight in $h_{i}(I)$
is at least $\bar{w}_{\gamma}$. Therefore, the total weight of the
discarded tuples is at most
\[n \bar{w}_{\gamma-t} \le {n w_\gamma \over
  (1+\eps')^{\lrceil{\log_{(1+\eps')}{(n/\eps')}}}} \le \eps' \bar{w}_\gamma \]
Since whenever $h_i$ is turned on by a given scenario, the sum of the
weights is at least $\bar{w}_\gamma$, we lose at most $\eps'$ fraction of the total weight
by discarding those tuples from $h_i(I)$. To see why we only lose an $\eps'$ fraction over all the hypotheticals
(instead of $\eps'k$), note that at most $n$ tuples will be discarded in the whole scenario,
hence the $n$ in the inequality $n \bar{w}_{\gamma-t} \le \eps' \bar{w}_\gamma$ can be
amortized over all the hypotheticals.
\end{proof}

\subsection{The Quantiles Query}\label{subsec:quant}
In this section, we study provisioning of the \quantileQuery query.
We again assume a relational schema with just one binary relation containing attributes identifier and weight.
For any instance $I$ and any tuple $x \in I$, we define the \emph{rank} of $x$ to be the number of tuples in $I$ that are smaller
than or equal to $x$ (in terms of the weights). The output of a \quantileQuery query with a given parameter $\phi \in (0,1]$ on an instance $I$ is the tuple
with rank $\lrceil{\phi \cdot \card{I}}$. 
 Finally, we say a tuple $x$ is a $(1 \pm \eps)$-approximation of a \quantileQuery query whose correct answer is
$y$, iff the rank of $x$ is a $(1\pm \eps)$-approximation of the rank of $y$.

Similar to the previous section, we first show that the \quantileQuery query admits no compact provisioning scheme for
exact answer and then provide a compact $\eps$-provisioning scheme for this query.

\begin{theorem}\label{thm:q-lower}
  Exact provisioning of the \quantileQuery query \emph{even on disjoint hypotheticals} requires sketches of size
  $\min(2^{\Omega(k)},\Omega(n))$ bits.
\end{theorem}

In the \quantileQuery query, the parameter $\phi$ may be given either already to the compression algorithm or only to the extraction
algorithm. The latter yields  an immediate lower bound of $\Omega(n)$, since by varying $\phi$ over $(0,1]$, one can effectively  ``reconstruct'' the original database.
However, we achieve a more interesting lower bound for the case when $\phi$ is given at to the compression algorithm (i.e., a fixed $\phi$ for all scenarios, e.g.,
setting $\phi = 1/2$ to find the \emph{median}). An important property of the lower bound for \quantileQuery is that, in contrast to all other lower bounds for numerical queries in this paper,
this lower bound holds even for disjoint hypotheticals\footnote{All other numerical queries that we study in this paper can be compactly provisioned for exact answer, when the hypotheticals are \emph{disjoint}.}.

\begin{proof}[Proof of Theorem~\ref{thm:q-lower}]
Assume we want to prove the lower
bound for any provisioning scheme for answering the median query
(\quantileQuery with fixed $\phi = 1/2$).

Let $N = 2^{k-1}$. We show how to encode a bit-string of length $N$ into a database $I$ with $n = \Theta(N)$ tuples
and a set of $k+1$ hypotheticals such that given provisioned sketch of the median query for this instance, one can recover any bit of this string with 
constant probability. Standard information-theoretic arguments then imply that the sketch size must have $\Omega(N) = 2^{\Omega(k)} = \Omega(n)$ bits.

For any vector $\bv{v} = (v_1,v_2,\ldots,v_N) \in
\set{0,1}^{N}$, define an instance $I_{\bv{v}}$ on a relational schema
with a binary relation $A$ whose second attribute is the weight. Let
$I_{\bv{v}} = L \cup M \cup R$, where $L = \set{A(x, 0)}_{x \in [N]}$,
$ M = \set{A(N + x, 2x + v_x)}_{x \in [N]}$, and $R = \set{A(2N + x,
  W)}_{x \in [2N]}$ (where $W$ is the largest possible value of the
weight). The weights of the tuples are ordered ``$L < M < R$''.  The
answers of all the scenarios will be in $M$: $L$ is the set of
``padding'' tuples which shifts the median towards $M$, and $R$ will
be divided into disjoint hypotheticals with different sizes (basically
size $2^{i-1}$ for the $i$-th hypothetical) where different scenarios
over such set of hypotheticals allow us to output \emph{all} the
tuples in $M$.

Formally, define the set $H$ of $(k + 1)$ hypotheticals, where for any
$i \in [k]$, $h_{i}(I_{\bv{v}}) =\set{A(2N +x, W)}_{x \in
  \set{2^{i-1}, \ldots, 2^i - 1}}$ with $2^{i-1}$ tuples, and $h_{k+1}
= L\cup M$. For any set $S' \subseteq [k]$, consider the scenario $S =
S' \cup \set{k+1}$. Let $\gamma = \ceil{\sum_{i \in S}
  \card{h_i(I_{\bv{v}})}/2}$. It is straightforward to verify that the
median tuple of $\app{I_{\bv{v}}}{S}$ is $A(\gamma, 2\gamma +
v_{\gamma})$. By varying $\sum_{i \in S} \card{h_i(I_{\bv{v}})}$ from
$1$ to $2N$ (using the fact that size of hypotheticals are different
powers of two), any tuple in $M$ will be outputted and the vector
$\bv{v}$ could be reconstructed. 
\end{proof}

Note that one can extend this lower bound, by using an approach similar to Theorem~\ref{thm:full-cover}, to  
provisioning schemes that are allowed a limited access to the original database after being given the scenario (see Section~\ref{sec:full-cover} for more details). 
We omit the details of this proof.

We now turn to prove the main theorem of this section, which argue the existence of a compact scheme for $\eps$-provisioning
the \quantileQuery. We emphasize that the approximation guarantee in the following theorem is \emph{multiplicative}. 

\begin{theorem}[\quantileQuery]\label{thm:quantiles}
For any $\eps,\delta > 0$, there exists a compact \edScheme for the
\quantileQuery query that creates a sketch of size
$\tilde{O}(k \eps^{-3} \log n \cdot (\log(n/\delta) + k )(\log{W} + k))$
bits.
\end{theorem}
We should note that in this theorem the parameter $\phi$ is only
provided in the extraction phase%
\footnote{We emphasize that we gave a lower bound for the easier case in terms
  of provisioning ($\phi$ given at compression phase and disjoint
  hypotheticals), and an upper bound for the harder case ($\phi$ given
  at extraction phase and overlapping hypotheticals).}.  
  Our starting
point is the following simple lemma first introduced
by~\cite{gupta2003counting}.

\begin{lemma}[\cite{gupta2003counting}]\label{lem:quantile}
For any list of \emph{unique} numbers $A = (a_1,\ldots,a_n)$ and
parameters $\eps,\delta > 0$, let $t =
\ceil{12\eps^{-2}\log{(1/\delta)}}$; for any \emph{target rank} $r >
t$, if we independently sample each element with probability $t/r$,
then with probability at least $1-\delta$, the rank of the $t$-th smallest
sampled element is a $(1 \pm \eps)$-approximation of $r$.
\end{lemma}

The proof of Lemma~\ref{lem:quantile} is an standard application of the
Chernoff bound and the main challenge for provisioning the \quantileQuery
query comes from the fact that hypotheticals overlap. We propose the
following scheme which addresses this challenge.

\cTextbox{Compression algorithm for the \quantileQuery query.} {Given an instance $I$, a set  $H$ of hypotheticals,
and two parameters $\eps, \delta > 0$, let $\eps' = \eps/5$, $\delta'
= \delta/3$, and $t = \ceil{12\eps'^{-2}(\log(1/\delta') + 2k +
  \log(n))}$.
\begin{enumerate}
    \item Create and \textbf{record} a \countSketch for $I$ and $H$
    with parameters $\eps'$ and $\delta'$.
    \item Let $\set{r_j = (1+\eps')^j}_{j=0}^{\lrceil{\log_{(1 +
          \eps')}{n}}}$. For each $r_j$, create the following
      sub-sketch individually.
    \item If $r_j \le t$, for each hypothetical $h_i$, \textbf{record}
      the $r_j$ smallest chosen tuples in $h_i(I)$. If $r_j > t$, for
      each hypothetical $h_i$, choose each tuple in $h_i(I)$ with
      probability $t/r_j$, and \textbf{record} the
      $\ceil{(1+3\eps')\cdot t}$ smallest tuples in a list $T_{i,j}$.
      For each tuple $x$ in the resulting list $T_{i,j}$,
      \textbf{record} its \emph{characteristics vector} for the set of
      the hypotheticals, which is a $k$-dimensional binary vector
      $(v_1, v_2, \dots, v_k)$, with value $1$ on $v_l$ whenever $x
      \in h_{l}(I)$ and $0$ elsewhere.
\end{enumerate}
}
\cTextbox{Extraction algorithm for the \quantileQuery query.} {Suppose we are given a scenario $S$ and a parameter $\phi \in
(0,1]$. In the following, the rank of a tuple always refers
to its rank in the sub-instance $\app{I}{S}$.
\begin{cEnumerate}
    \item Denote by $\tilde{n}$ the output of the \countSketch on
      $S$. Let $\tilde{r} = \phi \cdot \tilde{n}$, and find the index
      $\gamma$, such that $r_\gamma \le \tilde{r} < r_{\gamma+1}$.
    \item If $r_\gamma \le t$, among all the hypotheticals turned on
      by $S$, take the union of the recorded tuples and output the
      $r_\gamma$-th smallest tuple in the union.
    \item If $r_\gamma > t$, from each $h_i$ turned on by $S$, and
      each tuple $x$ recorded in $T_{i,\gamma}$ with a characteristic
      vector $(v_1, v_2, \dots, v_k)$, \emph{collect} $x$ iff for any
      $l < i$, either $v_l = 0$ or $h_l \notin S$. In other words, a
      tuple $x$ recorded by $h_i$ is taken only when among the
      hypotheticals that are turned on by $S$, $i$ is the smallest
      index s.t. $x \in h_i(I)$. We will refer to this procedure as
      the \emph{deduplication}.  Output the $t$-th smallest tuple
      among all the tuples that are collected.
\end{cEnumerate}
}

We call a sketch created by the above compression algorithm a
\quantilesSketch. It is straightforward to verify that the
size of \quantilesSketch is as stated in Theorem~\ref{thm:quantiles}. We now
prove the correctness of the above scheme.

\begin{proof}[Proof of Theorem~\ref{thm:quantiles}]
Given an instance $I$ and a set $H$ of hypotheticals , we prove that
with probability at least $1-\delta$, for every scenario $S$ and every
parameter $\phi \in (0,1]$, the output for the \quantileQuery query on
$\app{I}{S}$ is a $(1 \pm \eps)$-approximation. Fix a scenario $S$ and
a parameter $\phi \in (0,1]$; the goal of the extraction algorithm is
  to return a tuple with rank in range $(1 \pm \eps)$ of the
  \emph{queried rank} $\ceil{\phi\cdot\card{\app{I}{S}}}$. Recall that
  in the extraction algorithm, $\tilde{n}$ is the output of the
  \countSketch. Therefore, $\tilde{r} = \phi \tilde{n}$ is a $(1 \pm
  \eps')$ approximation of the queried rank, and $r_\gamma$ with
  $r_\gamma < \hat{r} \le (1+\eps')r_\gamma$ is a $(1 \pm 3\eps')$
  approximation of the rank. In the following, we argue that the tuple
  returned by the extraction algorithm has a rank in range $(1 \pm
  \eps')\cdot r_\gamma$, and consequently is a $(1 \pm
  \eps)$-approximation answer to the \quantileQuery query.

If $r_\gamma \le t$, the $r_\gamma$-th smallest tuple of $\app{I}{S}$
is the $r_\gamma$-th smallest tuple of the union of the recorded
tuples $T_{i,\gamma}$ for $i \in S$.  Therefore, we obtain a $(1 \pm
\eps)$ approximation in this case.

If $r_\gamma > t$, for any $i \in S$, define $\hat{T}_{i}$ to be the
list of all the tuples sampled from $h_i(I)$ in the compression algorithm
(instead of maintaining the $(1+3\eps')t$ tuples with smallest ranks).
Hence $T_{i,\gamma} \subseteq \hat{T}_i$.  If we perform the
deduplication procedure on the union of the tuples in $\hat{T}_i$ for
$i \in S$ and denote the resulting list $T^*$, then every tuple in
$\app{I}{S}$ has probability exactly $t/r_\gamma$ to be taken into
$T^*$ (for any tuple, only the appearance in the smallest index
$h_i(I)$ could be taken). Hence, by Lemma~\ref{lem:quantile}, with
probability at least $1 - \delta'/2^{k+\log(n)}$, the rank of the
$t$-th tuple of $T^*$ is a $(1 \pm \eps')$-approximation of the rank
$r_\gamma$.  In the following, we assume this holds.

The extraction algorithm does not have access to
$\hat{T}_{i}$'s. Instead, it only has access to the list
$T_{i,\gamma}$, which only contains the $(1 + 3\eps')t$ tuples of
$\hat{T}_{i}$ with the smallest ranks. We show that with high
probability, the union of $T_{i,\gamma}$ for $i \in S$, contains the
first $t$ smallest tuples from $T^*$. Note that for any $i \in S$, if
the largest tuple $x$ in $\hat{T}_{i} \cap T^*$ is in $T_{i,\gamma}$,
then all tuples in $\hat{T}_{i} \cap T^*$ are also in $T_{i,\gamma}$
(e.g. the truncation happens after the tuple $x$). Hence, we only need
to bound the probability that $x$ is truncated, which is equivalent to
the probability of the following event: more than $(1+3\eps')t$ tuples
in $h_i(I)$ are sampled in the compression phase for the rank
$r_\gamma$ (with probability $t/r_\gamma$).

Let $L_{i}$ be the set of all the tuples in $h_{i}(I)$ which are
smaller than $x$.  Rank of $x$ is less than or equal to the rank of
the $t$-th smallest tuple in $T^*$, which is upper bounded by
$(1+\eps')r_\gamma$. Hence, $\card{L_i} \leq (1+\eps)r_\gamma$. For
any $j \in h_i(I)$, define a binary random variable $y_j$, which is
equal to $1$ iff the tuple with rank $j$ is sampled and $0$ otherwise.
The expected number of the tuples that are sampled from $L$ is then $E[\sum_{j \in L} y_j] = \card{L}\cdot (t/r_\gamma) \leq (1+\eps')t$.

Using Chernoff bound, the probability that more than $(1+3\eps')t$
tuples from $L$ are sampled is at most ${\delta' \over 2^{2k +
    \log(n)}}$. If we take union bound over all the hypotheticals in
$S$, with probability at least $1 - {\delta' \over 2^{k + \log(n)}}$,
for all $h_i$, the largest tuple in $\hat{T}_i \cap T^*$ is in
$T_{i,\gamma}$, which ensures that the rank of the returned tuple is a
$(1 \pm \eps)$-approximation of the queried rank.

Finally, since there are only $n$
different values for $\phi \in (0,1]$ which results in different
  answers, applying union bound over all these $n$ different values of
  $\phi$ and $2^{k}$ possible scenarios, with probability at least $(1
  - 2\delta')$, the output of the extraction algorithm is a $(1 +
  \eps)$ approximation of the \quantileQuery query.  Since the failure
  probability of creating the \countSketch is at most $\delta'$, with
  probability $(1 - 3\delta') = (1-\delta)$ the \quantilesSketch successes.
\end{proof}

\paragraph{Extensions.}  By simple extensions of our scheme, many variations of the \quantileQuery query
can be answered, including outputting the rank of a tuple $x$, the percentiles (the rank of $x$ divided by the size of the input), or the tuple whose rank is $\Delta$ larger
than $x$, where $\Delta > 0$ is a given parameter. As an example,
for finding the rank of a tuple $x$, we can find the tuples with ranks approximately
$\set{(1+\eps)^l}$, $l \in [\ceil{\log_{(1+\eps)}{n}}]$, using the
\quantilesSketch, and among the found tuples, output the rank of the
tuple whose weight is the closest to the weight of $x$.

\subsection{The Linear Regression Query}\label{sec:reg}

In this section, we study provisioning of the \regQuery query (i.e., the $\ell_2$-regression problem), where 
the input is a matrix $\bvA_{n \times d}$ and a 
vector $\bvb_{n \times 1}$, and the goal is to output a vector $\bvx$ that
minimizes $\norm{\bvA \bvx - \bvb}$ ($\norm{\cdot}$ stands for the
$\ell_2$ norm). A $(1 + \eps)$-approximation of the \regQuery query is
a vector $\tilde{\bvx}$ such that $\norm{\bv{A}\tilde{\bvx} - \bv{b}}$
is at most $(1 + \eps) \min_\bvx\norm{\bv{A}\bvx - \bv{b}}$.

The input is specified using a relational schema $\Sigma$ with a $(d+2)$-ary relation $R$. Given an
instance $I$ of $\Sigma$ with $n$ tuples, we interpret the projection
of $R$ onto its first $d$ columns, the $(d+1)$-th column, and the $(d+2)$-column respectively as the matrix $\bv{A}$, the column vector $\bv{b}$,  and the identifiers for the tuples in $R$.  For
simplicity, we denote $I =(\bvA,\bvb)$, assume that the tuples are ordered,
and use the terms the $i$-th tuple of $I$ and the $i$-th row of
$(\bvA,\bvb)$ interchangeably.

\paragraph{Notation.} For any matrix $\bvM \in \IR^{n \times d}$, 
denote by $\bvM_{(i)}$ the $i$-th row of $\bvM$, and by
$\bv{U}_M \in \IR^{n \times \rho}$ (where $\rho$ is the rank of $\bvM$) the orthonormal matrix of the
column space of $\bv{M}$ (see~\cite{horn2012matrix} for more details). Given an instance $I = (\bvA, \bvb)$, and
$k$ hypotheticals, we denote for each hypothetical $h_i$ the
sub-instance $h_i(I) = (\bv{A}_{i}, \bv{b}_i)$. For any integer $i$,
$\bve_i$ denotes the $i$-th standard basis; hence, the $i$-th row of $\bvM$ can be written as $\bve_i^T
\bv{M}$.

The following theorem shows that the \regQuery query cannot be compactly provisioned for exact answers and hence, we will focus on $\eps$-provisioning.

\begin{theorem}\label{thm:reg-lower}
  Exact provisioning of the \regQuery query, \emph{even when the dimension is $d=1$}, requires sketches of
  size $\min(2^{\Omega(k)},\Omega(n))$ bits
\end{theorem}
\begin{proof}
  We show how to use the \regQuery query to check whether the sum of a
  list of values is equal to $0$, then, the same reduction used in
  Theorem~\ref{thm:neg-sum} will complete the proof.

  Given a list of values $a_1, a_2, \dots, a_n$, create an instance of
  $\ell_2$-regression problem with $d=1$ where $\bvA^T_{n \times
    1} = [1, 1, \dots, 1]$ and $\bvb^T_{n \times 1} = [a_1, a_2,
    \dots, a_n]$. We will show that $\sum_{i=1}^n a_i = 0$ iff
  $\arg\min_x{\norm{\bvA x - \bvb}^2}$ is $0$.
  \[\norm{\bvA x - \bvb}^2 = \sum_{i=1}^n (x - a_i)^2 = n x^2 -
  (2\sum_{i=1}^n a_i) x + \sum_{i=1}^n a_i^2\]

  This parabola reaches
  its minimum when $x = (\sum_{i=1}^n a_i)/n$ and the result follows.
\end{proof}

Before continuing, we remark that if hypotheticals are disjoint, the \regQuery query admits compact provisioning for exact answers, and hence
 the hardness of the problem again lies on the fact that hypotheticals overlap. To see this, consider the closed form solution:
\begin{alignat}{2}
\bvx_{opt} = (\bvA^T \bvA)^{\dagger} \bvA^T \bvb \label{eq:x-opt}
\end{alignat}
where $^{\dagger}$ denotes the \emph{Moore-Penrose pseudo-inverse} of
a matrix~\cite{horn2012matrix}.  Let $\bvA^T = [ \bvA_1^T, \bvA_2^T,
  \dots, \bvA_k^T]$ and $\bvb^T = [ \bvb_1^T, \bvb_2^T, \dots,
  \bvb_k^T]$, where for any $i \in [k]$, the $(\bvA_i, \bvb_i)$ pair
corresponds to the sub-instance of the $i$-th hypotheticals, i.e.,
$h_i(I)$. Then
\[\bvA^T \bvA = \sum_{i=1}^{k} \bvA_i^T \bvA_i \hspace{4mm} \bvA^T \bvb = \sum_{i=1}^{k} \bvA_i^T \bvb_i\]
Therefore, we can design a provisioning scheme that simply records the $d \times d$
matrix $\bvA_i^T \bvA_i$ and the $d$ dimensional vector $\bvA_i^T
\bvb_i$ for each hypothetical.  Then, for any given scenario $S
\subseteq [k]$, the extraction algorithm computes $\sum_{i \in S}
\bvA_i^T \bvA_i$ and $\sum_{i \in S} \bvA_i^T \bvb_i$, and obtains
$\bvx_{opt}$ by using Equation (\ref{eq:x-opt}).

We now turn to provide a provisioning scheme for the \regQuery query and prove the following theorem, which is 
the main contribution of this section. 

\begin{theorem}[\regQuery]\label{thm:regression}
For any $\eps,\delta >0$, there exists a compact \edScheme for the
\regQuery query that creates a sketch of size $\Ot(\eps^{-1}k^3 d\log(nW) (k+\log{1 \over \delta}))$
bits.
\end{theorem}

\paragraph{Overview.} Our starting point is a non-uniform sampling based approach (originally used for speeding up the
$\ell_2$-regression computation~\cite{sarlos2006improved}) which uses 
a small sample to accurately approximates the $\ell_2$-regression
problem. Since the probability of sampling a tuple (i.e., a row of the
input) in this approach depends on its relative importance which can vary dramatically when input changes, this approach is not directly applicable to our setting.

Our contribution is a \emph{two-phase sampling} based approach to
achieve the desired sampling probability distribution for \emph{any}
scenario. At a high level, we first sample and record a small number
of tuples from each hypothetical using the non-uniform sampling
approach; then, given the scenario in the extraction phase, we re-sample
from the recorded tuples of the hypotheticals presented in the
scenario. Furthermore, to rescale the sampled tuples (as needed in the
original approach), we obtain the exact sampling probabilities of the
recorded tuples by recording their relative importance in each
hypothetical.  Our approach relies on a monotonicity property of the
relative importance of a tuple when new tuples are added to the
original input.

\paragraph{RowSample.} We first describe the non-uniform sampling
algorithm.  Let $\probP$ be a probability distribution, and $r > 0$ be
an integer.  Sample $r$ tuples of $I = (\bvA, \bvb)$ with replacement
according to the probability distribution $\bvp$. For each sample, if
the $j$-th row of $\bvA$ is sampled for some $j$, rescale the row with
a factor $(1/\sqrt{r p_j})$ and store it in the sampling matrix
$(\tilde{\bvA},\tilde{\bvb})$.  In other words, if the $i$-th sample
is the $j$-th row of $I$, then $(\row{\tilde{\bvA}}{i}, \row{\tilde{\bvb}}{i}) = (\row{\bvA}{j},\row{\bvb}{j})/\sqrt{rp_j}$. We denote this procedure by $\rowsample(\bvA, \bvb, \bvp, r)$, and $(\tilde{\bvA}, \tilde{\bvb})$ is its output.  The $\rowsample$ procedure has the following property~\cite{sarlos2006improved} (see also~\cite{drineas2006sampling,woodruff2014sketching} for more details on introducing the parameter $\beta$).

\begin{lemma}[\cite{sarlos2006improved}]\label{lem:reg-sample}
  Suppose $\bv{A} \in \IR^{n \times d}$, $\bv{b} \in \IR^n$, and $\beta \in (0,1]$;
  $\probP$ is a probability distribution on $[n]$, and $r > 0$ is an integer.  Let
  $(\tilde{\bv{A}},\tilde{\bv{b}})$ be an output of
  $\rowsample(\bv{A},\bv{b}, \bvp,r)$, and $\tilde{\bvx} = \arg
  \min_\bvx\norm{\tilde{\bv{A}}\bvx - \tilde{\bv{b}}}$.

  If for all $i
  \in [n]$, $p_i \geq \beta
  \frac{\norm{\eiT {\bv{U}_A}}^2}{\sum_{j=1}^{n}\norm{e_j^T{\bv{U}_A}}^2}$
  , and $r =
    \Theta(\frac{d\log{d}\log(1/\delta)}{\eps\cdot\beta})$, then with
    probability at least $(1-\delta)$, 
  $\norm{\bv{A}\tilde{\bvx} - \bv{b}} \leq (1+\eps) \min_{\bvx}
        \norm{\bv{A}\bv{x} - \bv{b}}$.
\end{lemma}

The value $\norm{\eiT{\bv{U}_{A}}}^2$, i.e the square norm of the
$i$-th row of $\bv{U}_A$, is also called the \emph{leverage score} of
the $i$-th row of $\bv{A}$. One should view the leverage scores as the
``relative importance'' of a row for the $\ell_2$-regression problem
(see~\cite{mahoney2011randomized} for more details). Moreover, using
the fact that columns of $\bv{U}_A$ are orthonormal, we have
$\sum_{j}\norm{e_j^T{\bv{U}_{A}}}^2 = \rho$, where $\rho$ is the rank
of $\bvA$.

We now define our provisioning scheme for the \regQuery query, where
the compression algorithm performs the first phase of sampling
(samples rows from each hypothetical) and the extraction algorithm
performs the second (samples from the recorded tuples).

\cTextbox{Compression algorithm for the \regQuery query.} {Suppose we are given an instance $I =
(\bv{A},\bv{b})$ and $k$ hypotheticals with $h_i(I) = (\bvA_i,
\bv{b}_i)$ ($i \in [k]$).  Let $t =
\Theta({\eps^{-1}kd \log{d}\cdot(k+\log{(1/\delta)})})$, and define for each
$i\in[k]$ a probability distribution $\probPi$ as follows.  If the
$j$-th row of $\bvA$ is the $l$-th row of $\bvA_i$ (they correspond to
the same tuple), let $p_{i,j} =
{\norm{e_l^T{\bv{U}_{A_i}}}^2}/{\rho}$, where $\rho$ is the rank of $\bvA_i$. If $\row{\bvA}{j}$ does not
belong to $\bvA_i$, let $p_{i,j} = 0$. Using the fact that for every
$i \in [k]$, $\bv{U}_{A_i}$ is an orthonormal matrix, $\sum_{j=1}^{n}
p_{i,j} = 1$. \textbf{Record} $t$ independently chosen random
permutations of $[k]$, and for each hypothetical $h_i$, create a
sub-sketch as follows.
\begin{enumerate}
\item Sample $t$ tuples of $h_i(I)$ with replacement, according to
  the probability distribution $\bvp_i$.
\item For each of the sampled tuples, assuming it is the $j$-th tuple of
  $I$, \textbf{record} the tuple along with its sampling rate in each
  hypothetical, i.e., $\set{p_{i',j}}_{i' \in [k]}$.
\end{enumerate}
}

\cTextbox{Extraction algorithm for the \regQuery query.} {Given a scenario $S = \set{ i_1, \dots, i_s}$, we
will recover from the sketch a matrix $\tilde{\bv{A}}_{t \times d}$
and a vector $\tilde{\bvb}_{t \times 1}$. For $l = 1$ to $t$:
\begin{cEnumerate}
\item Pick the $l$-th random permutation recorded in the sketch.  Let $\gamma$ be the first value in this permutation that appears in $S$.
\item If $(\bva, b)$ is the $l$-th tuple sampled by the hypothetical
  $h_\gamma$, which is the $j$-th tuple of $I$, let $q_j = \sum_{i\in
  S}p_{i,j}/{\card{S}}$, using the recorded sampling rates. \label{step:qj}
\item Let $(\tilde{\bv{A}}_{(l)}, \tilde{\bvb}_{(l)}) = ( \bva ,b)
  /\sqrt{t q_j}$.  Return $\tilde{\bv{x}} = \arg{\min}_{\bv{x}}
  \norm{\tilde{\bv{A}}\bv{x} - \tilde{\bv{b}}}$ (using any standard
  method for solving the $\ell_2$-regression problem).
\end{cEnumerate}
}

We call a sketch constructed above a \regressionSketch.  In order to 
show the correctness of this scheme, we need the
following lemma regarding the monotonicity of leverage
scores.

\begin{lemma}[Monotonicity of  Leverage Scores]\label{lem:ls-monotone}
 Let $\bv{A} \in \IR^{n \times d}$ and $\bv{B} \in \IR^{m \times d}$
 be any matrix. Define matrix $\bv{C} \in \IR^{(n+m)\times d}$ by
 \emph{appending} rows of $\bv{B}$ to $\bv{A}$, i.e., $\bv{C}^T =
      [\bv{A}^T, \bv{B}^T]$.  For any $i \in [n]$, if $L_i$ is the
      leverage score of $\row{\bvA}{i}$ and $\hat{L}_i$ is the
      leverage score of $\row{\bv{C}}{i}$, then $L_i \geq \hat{L}_i$.
\end{lemma}
Before providing the proof, we note that Lemma~\ref{lem:ls-monotone} essentially claims that adding more rows
to a matrix $\bvA$ can only reduce the importance of any point
originally in $\bvA$.  Note that this is true even when the matrix
$\bv{C}$ is formed by arbitrarily combining rows of $\bv{B}$ and
$\bvA$ (rather than appending at the end).  
\begin{proof}[Proof of Lemma~\ref{lem:ls-monotone}]
We use the following characterization of the leverage scores (see~\cite{cohen2014uniform} for a proof).
 For any matrix $\bvM \in \IR^{n \times
  d}$ and $i \in [n]$, the leverage score of $\row{\bvM}{i}$ is equal 
  \[
  \min_{\bvM^{T} \bx = \bvM_{(i)}} \norm{\bx}^2
  \]
Now, fix an $i \in [n]$, and let $\bvx^* = \arg\min_{\bvA^{T} \bx =
  \bvA_{(i)}} \norm{\bx}^2$. By the above characterization of leverage
scores, $L_{i} = \norm{\bvx^*}^2$.  Denote by $\bv{z}$, the vector in
$\IR^{n+m}$ where $\bv{z}^T = [ {\bvx^*}^T,~\bv{0}]$. Then, 
\[
\bv{C}^T \cdot \bv{z} = \begin{bmatrix} \bvA^T,~\bv{B}^T \end{bmatrix} \cdot \bv{z} = \bvA^T{\bvx^*} = \bvA_{(i)} = \bv{C}_{(i)}
\]
where the last two equalities are by the definitions of $\bvx^*$ and
$\bv{C}$, respectively. Therefore, 
\[
\hat{L}_i = \min_{\bv{C}^{T} \bx = \bv{C}_{(i)}} \norm{\bx}^2 \leq \norm{\bv{z}}^2 = \norm{\bvx^*}^2 = L_{i}
\]
where the inequality holds because $\bv{z}$ satisfies $\bv{C}^{T} \bv{z} = \row{\bv{C}}{i}$.
\end{proof}

\begin{proof}[Proof of Theorem~\ref{thm:regression}]
Fix a scenario $S$ and let $\app{I}{S} = (\hat{\bvA}, \hat{\bvb})$.
It is straightforward to verify that, for any step $l \in [t]$, $q_j$
(in line (\ref{step:qj}) of the extraction algorithm) is the probability
that the $j$-th tuple of $I$ is chosen, if the $j$-th tuple belongs to
$\app{I}{S}$. Hence, assuming $\bvp'$ is the probability distribution
defined by the $q_j$s on rows of the $\app{I}{S}$, the extraction
algorithm implements $\rowsample(\hat{\bvA}, \hat{\bvb}, \bvp',
t)$. 

We will show that $
q_j \ge {\norm{e_l^T {\bv{U}_{\hat{A}}}}^2}/ {k \hat{\rho}} 
$,
where the $l$-th row of $\hat{\bvA}$ is the $j$-th row of $\bvA$, and
$\hat{\rho}$ is the rank of $\hat{A}$. 
Then, by Lemma~\ref{lem:reg-sample} with $\beta$ set to $1/k$, with probability
at least $1 - \frac{\delta}{2^{k}}$, $\norm{\hat{\bvA}\tilde{\bvx} -
  \hat{\bvb}}$ is at most $(1+\eps) \min_{\bvx} \norm{\hat{\bvA}\bv{x}
  - \hat{\bvb}}$; hence, the returned vector $\tilde{\bvx}$ is a
$(1+\eps)$-approximation.  Applying a union bound over all $2^k$
scenarios, with probability at least $(1 - \delta)$, our scheme
$\eps$-provisions the \regQuery query.

We now prove that $q_j \ge {\norm{e_l^T {\bv{U}_{\hat{A}}}}^2}/ {k \hat{\rho}}$. Denote by
$L_{i,j}$ (resp. $L_{S,j}$) the leverage score of the $j$-th tuple of
$I$ in the matrix $\bvA_i$ (resp. $\hat{\bvA}$).  Further, denote by
$\rho_i$ the rank of $\bvA_i$.  Consequently, $p_{i,j} =
L_{i,j}/\rho_i$, and our goal is to show $q_j \ge L_{S,j}/(k \hat{\rho})$.
Pick any $i^* \in S$ where $h_{i^*}(I)$ contains the $j$-th tuple of
$I$, then:
\begin{cEqnarray}
  q_{j} ={\sum_{i \in S}p_{i,j} \over s}
  \geq \frac{p_{i^*,j}}{s}  \geq \frac{L_{i^*,j}}{k \rho_i}
  \geq  \frac{L_{S,j}}{k \hat{\rho}}
\end{cEqnarray}
For the last inequality, since $\bvA_i$ is a sub-matrix of
$\hat{\bvA}$, $\rho_i \leq \hat{\rho}$ and the leverage score decreases
due to the monotonicity (Lemma~\ref{lem:ls-monotone}).

To conclude, the probabilities will be stored with precision $1/n$
(hence stored using $O(\log{n})$ bits each), and the size of the sketch
is straightforward to verify.
\end{proof}

\section{Complex Queries}\label{sec:mix}
We study the provisioning of
queries that combine logical components (relational algebra and Datalog),
with grouping and with the numerical queries that we studied in
Section~\ref{sec:num-qry}.

We start by defining a class of such queries and their semantics
formally. For the purposes of this paper, a \emph{complex query} is a
triple \compose where $Q_L$ is a relational algebra or Datalog query
that outputs some relation with attributes $\bar{A}\bar{B}$ for some
$\bar{B}$, $G_{\bar{A}}$ is a \emph{group-by} operation applied on the
attributes $\bar{A}$, and $Q_{N}$ is a numerical query that takes as
input a relation with attributes $\bar{B}$.  For any input $I$ let
$P=\Pi_{\bar{A}}(Q_L(I))$ be the $\bar{A}$-relation consisting of all
the distinct values of the grouping attributes.  We call the size
of $P$ the \emph{number of groups} of the complex query.
For each tuple $\bar{u}\in P$, we define $\Gamma_{\bar{u}}=\{\bar{v}\:|\:
\bar{u}\bar{v}\in Q_L(I)\}$.  Then, the output of the complex query
\compose is a set of tuples $\{\langle
\bar{u},Q_N(\Gamma_{\bar{u}})\rangle \:|\: \bar{u}\in P\}$.

\subsection{Positive, Non-Recursive Complex Queries}

In the following, we give compact provisioning results for the case where the
logical component is a \emph{positive relational algebra} (i.e., SPJU)
query.  It will be convenient to assume a different, but equivalent,
formalism for these logical queries, namely that of \emph{unions of
conjunctive queries} (UCQs)\footnote{Although the translation of an SPJU
query to a UCQ may incur an exponential size blowup~\cite{AHV}, in
this paper, query (and schema) size are assumed to be
constant. In fact, in practice, SQL queries often present
  with unions already at top level.}. We review quickly the definition
of UCQs.  A \emph{conjunctive query} (CQ) over a relational schema
$\Sigma$ is of the form $\mathit{ans}(\overline{x})\dtlg R_1(\overline{x}_1),\ldots,R_b(\overline{x}_b)$,
where atoms $R_1, \ldots, R_b \in \Sigma$,
and the \emph{size} of a CQ is defined to be the number of atoms in its body
(i.e., $b$).
A union of conjunctive query (UCQ) is a finite union of some CQs
whose heads have the same schema.

In the following theorem, we show that for any complex query, where
the logical component is a positive relational algebra query, 
compact provisioning of the
numerical component implies compact provisioning of the complex query
itself, provided the number of groups is not too large.

\begin{theorem}\label{thm:mix}
For any complex query \compose where $Q_L$ is a \UCQ, if the numerical
component $Q_N$ can be compactly provisioned (resp. compactly
$\eps$-provisioned), and if the number of groups
%Section~\ref{sec:prob}, we assume that 
is bounded by $\poly(k,\log{n})$,
then the query \compose can also be compactly provisioned
(resp. compactly $\eps$-provisioned with the same parameter $\eps$).
\end{theorem}

\begin{proof}
Suppose $Q_N$ can be compactly provisioned (the following proof also
works when $Q_N$ can be compactly $\eps$-provisioned).  Let $b$ be the
maximum size of the conjunctive queries in $Q_L$.  Given an input
instance $I$ and a set $H$ of $k$ hypotheticals, we define a new
instance $\hat{I} = Q_L(I)$ and a set $\hat{H}$ of $O(k^b)$ new
hypotheticals as follows. For each subset $S \subseteq [k]$ of size at
most $b$ (i.e., $\card{S} \leq b$), define a hypothetical
$\hat{h}_S(\hat{I}) = Q_L(\app{I}{S})$ (though $S$ is not a number, we
still use it as an index to refer to the hypothetical $\hat{h}_S$).

By our definition of the semantics of complex queries, the group-by
operation partitions $\hat{I}$ and each $\hat{h}_S$ into
$p=|\Pi_{\bar{A}}(\hat{I})|$ sets. We treat each group individually, and
create a sketch for each of them. To simplify the notation, we still
use $\hat{I}$ and $\hat{H}$ to denote respectively the portion of the
new instance, and the portion of each new hypothetical that correspond
to, without loss of generality, the first group.  In the following, we
show that a compact provisioning scheme for $Q_N$ with input $\hat{I}$
and $\hat{H}$ can be adapted to compactly provision \compose for the
first group. Since the number of groups $p$ is assumed to be
$\poly(k,\log{n})$, the overall sketch size is still
$\poly(k,\log{n})$, hence achieving compact
provisioning for the complex query.

Create a sketch for $Q_N$ with input $\hat{I}$ and $\hat{H}$. For any
scenario $S \in [k]$ (over $H$), we can answer the numerical query
$Q_N$ using the scenario $\hat{S}$ (over $\hat{H}$) where $\hat{S} =
\set{S' \mid S' \subseteq S ~\&~ \card{S'} \le b}$.  To see
this, we only need to show that the input to $Q_N$ remains the same,
i.e., $Q_L(\app{I}{S})$ is equal to $\app{\hat{I}}{{\hat{S}}}$.  Each
tuple $t$ in $Q_L(\app{I}{S})$ can be derived using (at most) $b$
hypotheticals. Since any subset of $S$ with at most $b$ hypotheticals
belongs to $\hat{S}$, the tuple $t$ belongs to
$\app{\hat{I}}{\hat{S}}$. On the other hand, each tuple $t'$ in
$\app{\hat{I}}{\hat{S}}$ belongs to some $\hat{h}_{S'}$ where $S' \in
S$, and hence, by definition of $\hat{h}_{S'}$, the tuple $t$ is also
in $Q_L(\app{I}{S})$. Hence, $Q_L(\app{I}{S}) =
\app{\hat{I}}{{\hat{S}}}$.

Consequently, any compact provisioning scheme for $Q_N$ can be adapted
to a compact provisioning scheme for the query \compose.
\end{proof}

Theorem~\ref{thm:mix} further motivates our results in
Section~\ref{sec:num-qry} for numerical queries as they can be
extended to these quite practical queries. Additionally, as an
immediate corollary of the proof of Theorem~\ref{thm:mix}, we obtain that any
boolean \UCQ (i.e., any UCQ that outputs a boolean answer rather than
a set of tuples) can be compactly provisioned.

\begin{corollary}\label{cor:cq-upper}
Any boolean UCQ can be compactly provisioned using sketches of size
$O(k^{b})$ bits, where $b$ is the maximum size of each CQ.
\end{corollary}

\begin{remark}
Deutch \etal~\cite{deutch2013caravan}
introduced query provisioning from a practical perspective and
proposed \emph{boolean
  provenance}~\cite{ImielinskiL84,GreenKT07,Green11,2011Suciu} as a
way of building sketches. This technique can also be used for
compactly provisioning boolean UCQs.
\end{remark}
\begin{proof}[Proof Sketch]
Given a query $Q$, an instance $I$ and a set $H$
of hypotheticals we compute a small sketch in the form of a boolean
provenance expression.

Suppose that each tuple of an instance $I$ is annotated with a
distinct provenance token. The provenance annotation of the answer to
a UCQ is a monotone DNF formula $\Delta$ whose variables are these
tokens. Crucially, each term of $\Delta$ has fewer than $b$ tokens
where $b$ is the size of the largest CQ in $Q$.  Associate a boolean
variable $x_i$ with each hypothetical $h_i\in H,~i \in
[k]$. Substitute in $\Delta$ each token annotating a tuple $t$ with
the disjunction of the $x_i$'s such that $t\in h_i(I)$ and with false
otherwise. The result $\Delta'$ is a DNF on the variables
$x_1,\ldots,x_k$ such that each of its terms has at most $b$
variables; hence the size of $\Delta'$ is $O(k^b)$. Now $\Delta'$ can
be used as a sketch if the extraction algorithm sets to true the
variables corresponding to the hypotheticals in the scenario and to
false the other variables.
\end{proof}

\begin{remark}
A similar approach based on rewriting boolean provenance annotations,
which are now general DNFs, can be used to provision \ucqs\ under
disjoint hypotheticals. The disjointness assumption insures that
negation is applied only to single variables and the resulting DNF
has size $O((2k+1)^b)=O(k^b)$.
\end{remark}

We further point out that the exponential dependence of the sketch
size on the query size (implicit) in Theorem~\ref{thm:mix} and
Corollary~\ref{cor:cq-upper} cannot be avoided even for CQs.

\begin{theorem}\label{thm:cq-lower}
 There exists a boolean conjunctive query $Q$ of size $b$ such that provisioning $Q$ requires sketches of size $\min(\Omega(k^{b-1}),\Omega(n))$ bits.
\end{theorem}

\begin{proof}[Proof of Theorem~\ref{thm:cq-lower}] 

Consider the following boolean conjunctive query $\qexp$ defined over
a schema with a unary relation $A$ and a $(b-1)$-ary relation $B$.
\begin{alignat*}{2}
	\qexp \equiv ans() \dtlg A(x_1), A(x_2), \ldots, A(x_{b-1}), \\
	B(x_1,x_2,\ldots,x_{b-1})
\end{alignat*}
We show how to encode a bit-string of length $N:= {k-1 \choose b-1}$ into a database $I$ with $n = \Theta(N)$ tuples
and a set of $k$ hypotheticals such that given provisioned sketch of $\qexp$, one can recover any bit of this string with 
constant probability. Standard information-theoretic arguments then imply that the sketch size must have $\Omega(N) = \Omega(k^{b-1}) = \Omega(n)$ bits.

Let $(S_1,\ldots,S_N)$ be a list of all
subsets of $[k-1]$ of size $b-1$.  For any vector $\bvv \in \set{0,1}^N$, define the instance $I_{\bvv} = \set{A(x)}_{x \in [k-1]} \cup
\set{B(S_y)}_{v_y = 1}$ (this is slightly abusing the notation:
$B(S_y) = B(x_1, \ldots, x_{b-1})$ where $\set{x_1, \ldots, x_{b-1}} =
S_y$), and a set $\set{h_i}_{i\in[k]}$ of $k$ hypotheticals, where for
any $i \in [k-1]$, $h_i(I) = \set{A(i)}$ and $h_k(I) =
\set{B(S_y)}_{v_y = 1}$. To compute the $i$-th entry of $\bvv$, we can extract the answer to the scenario $S_{i} \cup \set{k}$ from the sketch and output $1$ iff the answer of the query is true.

To see the correctness, $v_{i} = 1$ iff $B(S_{i}) \in h_k(I)$ iff
$\qexp(\app{I}{S_{i} \cup \set{k}})$ is true.
\end{proof}

Note that one can extend this lower bound, by using an approach similar to 
Theorem~\ref{thm:full-cover}, to  
provisioning schemes that are allowed a limited access to the original database after being given the scenario (see Section~\ref{sec:full-cover} for more details). 
We omit the details of this proof.

\subsection{Adding Negation, Recursion, or HAVING} 

It is natural to ask $(a)$ if Theorem~\ref{thm:mix} still holds when adding \emph{negation} or \emph{recursion} to the query $Q_L$
(i.e. \emph{\UCQ with negation} and \emph{recursive Datalog},
respectively), and $(b)$ whether or not it is possible to provision
queries in which logical operations are done \emph{after} numerical
ones. A typical example of a query in part $(b)$ is a selection on
aggregate values specified by a HAVING clause.
Unfortunately, the answer to both questions is
negative. 

We first show that the answer to the question $(a)$ is negative, i.e., there exists a
\emph{boolean conjunctive query with negation} $(\cq)$ and a recursive
Datalog query which require sketches of exponential size (in $k$) for
any provisioning scheme. Formally,

\begin{theorem}\label{thm:ucq-dl-lower}
  Exact provisioning of $(i)$ boolean CQ with negation or $(ii)$ recursive Datalog
  (even without negation) queries requires sketches of size
  $\min(2^{\Omega(k)},\Omega(n))$.
\end{theorem}

Recall that a $\cq$ is a rule of the form
\begin{eqnarray*}
\mathit{ans}() \dtlg L_1,\ldots L_b
\end{eqnarray*}
where the $L_i$'s are positive or negative literals (atoms or negated
atoms) that is subject to \emph{range-restriction}: every variable
occurring in the rule appears in at least one positive literal.  We
define the size of such a query to be the number of literals.

Define the following boolean $\cq$ over a schema with two unary relation symbols, named $A$ and $B$:
\begin{eqnarray*}
 \qnsub() \dtlg A(x),\neg B(x)
\end{eqnarray*}
This query returns true on $I$ iff there exists some $x$ where $A(x)
\in I$ and $B(x) \notin I$. Intuitively, if we view $A$, $B$ as
subsets of the active domain of $I$, it is querying whether or not
``$B$ is a subset of $A$''.  We use a reduction from the $\fullcover$
problem to prove the lower bound for $\qnsub$.

\begin{proof}[Proof of Theorem~\ref{thm:ucq-dl-lower}, Part (I)]
  Suppose we are given a collection $\set{S_1,\ldots,S_k}$ of subsets of $[n]$ and we want to solve the $\fullcover$ problem using a provisioning scheme for $\qnsub$.
  Create the following instance $I$ for the schema $\Sigma = \set{A, B}$, where for any $x \in [n]$, $A(x), B(x) \in I$. Define the
  set of of hypotheticals $H = \set{h_1,h_2,\ldots,h_{k+1}}$, where for any $i \in [k]$, $h_i(I)= \set{B(x)~|~x \in S_{i}}$ and $h_{k+1}(I) = \set{A(x)~|~x \in [n]}$.

  For any set $\hat{S} = \set{i_1,\ldots,i_s} \subseteq [k]$, under the scenario
  $S = \hat{S} \cup \set{k+1}$, $\qnsub(\app{I}{S})$ is true iff there exists $x \in [n]$ s.t. $B(x) \notin \app{I}{S}$ which is equivalent to $[n]\not\subseteq S_{i_1} \union \ldots \union S_{i_s}$.
  Therefore any provisioning scheme of $\qnsub$ solves the $  \fullcover$ problem, and it follows from Theorem~\ref{thm:full-cover} that the sketch of such scheme must have size $\min(2^{\Omega(k)}, \Omega(n))$ bits.
\end{proof}

While adding negation extends UCQs in one direction, adding recursion
to UCQ also results in another direction. This leads to
recursive Datalog (without negation). Consider the \stconQuery query, which
captures whether there is a path from the vertex $s$ to the vertex
$t$:
\begin{cEqnarray}
\mathit{ans}() & \dtlg  & T(\mathtt{t}) \\
     T(y)    & \dtlg  & E(x,y), T(x)\\
     T(\mathtt{s})
\end{cEqnarray}
To prove a lower bound for the \stconQuery query, we again use a
reduction from the $\fullcover$ problem, where in our construction,
the only path from $s$ to $t$ has $n$ edges, and hence only when all
the $n$ edges are presented in a scenario $s$ will be connected to
$t$.

 \begin{proof}[Proof of Theorem~\ref{thm:ucq-dl-lower}, Part (II)]
  Suppose we are given a collection $\set{S_1,\ldots,S_k}$ of subsets of $[n]$ and we want to
  solve the $\fullcover$ problem using a provisioning scheme for \stconQuery. Consider a graph $G(V,E)$ with vertex set
  $V = \set{(s = )v_0, v_1, v_2, \dots, v_n( = t)}$, and edges
  $E(v_{j-1}, v_j)$, for all $j \in [n]$. The edge set $E$ of the graph $G$ is the input $I$ to the provisioning scheme. Define the hypotheticals $H
  = \set{h_1, h_2, \dots, h_k}$ where $h_i(I) = \set{E(v_{j-1},
    v_j)}_{j \in S_i}$, for all $i \in [k]$.

  For any scenario $S =
  \set{i_1,\ldots,i_s} \subseteq [k]$, $T(t)$ is true (i.e., $s$ is
  connected to $t$) in $\app{I}{S}$ iff all $n$ edges are in
  $\app{I}{S}$, which is equivalent to $S_{i_1} \union \ldots \union
  S_{i_s} = [n]$.  Therefore, any provisioning scheme for the
  \stconQuery query solves the $\fullcover$ problem, and it follows
  from Theorem~\ref{thm:full-cover} that the sketch of such scheme must
  have size $\min(2^{\Omega(k)}, \Omega(n))$ bits.
\end{proof}

Showing a negative answer to the question $(b)$ in the beginning of this section is very easy. As we already showed in Section~\ref{sec:num-qry} (see, e.g., Theorem~\ref{thm:csa-lower}), there are numerical queries that do not
admit compact provisioning for \emph{exact} answer. One can simply verify that each of those queries can act as a counter example for question $(b)$ by considering HAVING clauses that test the equality of the answer to the numerical part against an exact answer (e.g. testing whether the answer to \countQuery is $n$ or not).

%%
%%It turns out that Theorem~\ref{thm:mix} fails if we add
%%\emph{negation} or (separately) \emph{recursion} to the query $Q_L$. 
%%Exponential lower bounds also hold (see Appendix~\ref{app:mix-ext} again)
%%for queries in which logical operations are performed
%%\emph{after} numerical operations (e.g., a HAVING clause). We now describe these results in more details. 
%%

\section{Comparison With a Distributed Computation Model}\label{sec:dis}

Query provisioning bears some resemblance to the
 following distributed computation model: $k$ sites want to jointly
 compute a function, where each site holds only a portion of the
 input.  The function is computed by a \emph{coordinator}, who
 receives/sends data from/to each site. The goal is to design
 protocols with small amount of communication between the sites and
 the coordinator
 (see~\cite{cormode2010optimal,woodruff2012tight},
and references therein).
In what follows, we highlight some key similarities and
 differences between this distributed computation
model and our query provisioning framework.  

In principle, any protocol where data are only sent from the sites to
the coordinator (i.e., \emph{one-way} communication), can be adapted
into a provisioning scheme with a sketch of size proportional to the
size of the transcript of the protocol.  To see this, view each
hypothetical as a site and record the transcripts as the sketch.
Given a scenario $S$, the extraction algorithm acts as the coordinator
and computes the function from the transcripts of the hypotheticals in
$S$.  The protocol for counting the number of distinct elements in the
distributed model introduced in~\cite{cormode2011algorithms} is an
example (which in fact also uses the streaming algorithm
from~\cite{bar2002counting} as we do in Lemma~\ref{lem:stream-count}).

However, protocols in the distributed model usually involve back and
forth (\emph{two-way}) communication between the sites and the
coordinator (e.g., the algorithms
of~\cite{yi2013optimal,CormodeMuthukrishnanZhuang06} for estimating
quantiles in the distributed model).  In general, the power of the
distributed computation model with two-way communication
is \emph{incomparable} to the query provisioning framework.
Specifically, for $k$ sites/hypotheticals (even when the input is
partitioned, i.e., no overlaps), there are problems where a protocol with $\poly(k)$-bit
transcripts exists but provisioning requires $2^{\Omega(k)}$-bit
sketches. Conversely, there are problems where provisioning can be
done using $\poly(k)$-bit sketches but any distributed protocol
requires $2^{\Omega(k)}$-bit transcripts. We make these observations precise
below.

Another source of difference between our techniques and the ones used in the distributed computation model is 
the typical assumption in the latter that the input is \emph{partitioned} among the sites in the distributed computation model (i.e. no overlaps; see, for instance, 
\cite{yi2013optimal,cormode2010optimal,woodruff2012tight}). This assumption is in contrast to the hypotheticals with unrestricted overlap that we handle in our query provisioning framework. A notable exception to the no-overlaps 
assumption is the study of quantiles (along with various other aggregates) by~\cite{CormodeMuthukrishnanZhuang06}. 
However, as stated by the authors, they benefit from the coordinator sharing a summary of the
whole data distribution to individual sites, which requires a back and forth communication between the sites and coordinator. As such the results of~\cite{CormodeMuthukrishnanZhuang06} can not be directly translated into a 
provisioning scheme. We also point out that the approximation guarantee we obtain for the \quantileQuery query is stronger 
than the result of~\cite{CormodeMuthukrishnanZhuang06} (i.e. a relative vs additive approximation). 

Finally, we use two examples to establish an exponential separation between the
minimum sketch size in the provisioning model and the minimum
transcript length in the distributed computation model, proving a formal separation between the two models.

Before describing the examples, we should note that,
in the following, we assume that input tuples
are made distinct by adding another column which contains
\emph{unique identifiers}, but the problems and the operations will
only be defined over the part of the tuples without the identifiers. 
For instance, `two sets of tuples are disjoint' means `\emph{after removing the
  identifiers of the tuples,} the two sets are disjoint'.

\paragraph{A problem with poly-size sketches and exponential-size transcripts.}
The following \Setdis problem is well-known to require transcripts of
$\Omega(N)$ bits for any protocol~\cite{Razborov92,bar2002information}: Alice is
given a set $S$ and Bob is given a set $T$ both from the universe
$[N]$, and they want to determine, with success probability at least
$2/3$, whether $S$ and $T$ are disjoint. Similarly, if each of the $k$
sites is given a set and they want to determine whether all their sets
are disjoint, the size of the transcript is also lower bounded by
$\Omega(N)$ bits. If we let $N = 2^k$, the distributed model requires
$\Omega(2^k)$ bits of communication for solving this problem.

However, to provision the \Setdis query (problem), we only need to record
the pairs of hypotheticals whose sets are not disjoint (hence $O(k^2)$
bits). The observation is that if a collection of sets are not disjoint,
there must exist two sets that are also not disjoint, which can be
detected by recording the non-disjoint hypothetical pairs.

\paragraph{A problem with exponential-size sketches and poly-size transcripts.}
Consider a relational schema with two unary relations $A$ and
$B$. Given an instance $I$, we let $a = \sum_{A(x)\in I} 2^x$; then, the
problem is to determine whether $B(a) \in I$ or not. Intuitively, $A$
`encodes' the binary representation of a value $a$, and the problem is
to determine whether $a$ belongs to `the set of values in $B$'.

Let $n = 2^k$, and $I = \set{A(x)_{x \in [k]}} \cup \set{B(y)}_{y \in
  [n]}$. Using a similar construction as
Lemma~\ref{lem:strict-full-cover}, one can show that provisioning this
problem requires a sketch of size $2^{\Omega(k)}$, even when the input
instance is guaranteed to be a subset of $I$ (i.e., the largest value
of $A$ is $k$).

However, in the distributed model, there is a protocol using a
transcript of size $\poly(k)$ bits: every site sends its tuple in $A$
to the coordinator; the coordinator computes $a = \sum_{A(x)} 2^x$ and
sends $a$ to every site; a site response $1$ iff it contains the tuple
$B(a)$, and $0$ otherwise.

\section{Conclusions and Future Work} \label{sec:conc}

In this paper, we initiated a formal framework to study compact
provisioning schemes for relational algebra queries, statistics/analytics
including quantiles and linear regression, and complex queries. We considered
provisioning for exact as well as approximate answers, and established
upper and lower bounds on the sizes of the provisioning sketches.

The queries in our study include quantiles and linear regression
queries from the list of in-database analytics highlighted
in~\cite{HellersteinRSWFGNWFLK12}. This is only a first step and the
study of provisioning for other core analytics problems, such as variance
computation, $k$-means clustering,
logistic regression, and support vector machines is of interest. 

Another direction for future research is the study of queries
in which \emph{numerical} computations follow each other (e.g., when the
linear regression training data is itself the result of aggregations).
Yet another direction for future research is an extension of our model to
allow other kinds of hypotheticals/scenarios as discussed
in~\cite{deutch2013caravan} that are also of practical interest.  For
example, an alternative natural interpretation of hypotheticals 
is that they represent
tuples to be \emph{deleted} rather than
retained. Hence the application of a scenario $S \subseteq [k]$ to $I$
becomes $\app{I}{S}=I\setminus (\bigcup_{i\in S} h_i(I))$. Using our lower bound techniques,
one can easily show that even simple queries like count or sum cannot be approximated to within any multiplicative 
factor under this definition. Nevertheless, it will be interesting to identify query classes that admit compact provisioning in the delete model or alternative natural models.

\clearpage
\bibliographystyle{plain}% the recommended bibstyle
\bibliography{ref}

\clearpage
\appendix
\section{Tools from Information Theory}\label{app:info}
We use basic concepts from information theory in our lower bound proof in Section~\ref{sec:full-cover}. For a broader introduction to the field and proofs of the claims in this section, we refer the reader
 to the excellent text by Cover and Thomas~\cite{ITbook}. 

In the following, we denote random variables using capital bold-face letters. 
We denote the \emph{Shannon Entropy} of a random variable $\bA$ by $H(\bA)$ and the \emph{mutual information} of two random variables $\bA$ and $\bB$ by 
$I(\bA;\bB) = H(\bA) - H(\bA \mid \bB) = H(\bB) - H(\bB \mid \bA)$. For a real number $0 \leq x \leq 1$, we further use:
$H_2(x) = x\log{\frac{1}{x}} + (1-x)\log{\frac{1}{1-x}}$ to denote the binary entropy function. 
Finally, we use $\supp{A}$ to denote the support of the random variable $\bA$ and $\card{\bA}:= \log{\card{\supp{A}}}$. 

We use the following basic properties of entropy and mutual information. 
\begin{claim}\label{clm:it-facts}
	Suppose $\bA$, $\bB$, and $\bC$ are random variables and $f : \supp{A} \mapsto \supp{B}$ is a function;
	\begin{enumerate}
		\item \label{part:uniform} $0 \leq H(\bA) \leq \card{A}$; $H(\bA) = \card{A}$ iff $\bA$ is uniformly distributed over its support. 
		\item \label{part:cond-reduce} \emph{Conditioning reduces entropy}, i.e, $H(\bA \mid \bB) \leq H(\bA)$. 
		\item \label{part:info-zero} $I(\bA ; \bB) \geq 0$; equality holds iff $\bA$ and $\bB$ are \emph{independent}. 
		\item \label{part:sub-additivity} \emph{Subadditivity of entropy}:  $H(\bA,\bB) \leq H(\bA) + H(\bB)$. 
		\item \label{part:fano} \emph{(Fano's inequality)} For a \emph{binary} random variable $\bB$, and a function $f$ that predicts $\bB$ based
		on $\bA$, if $\Pr(f(\bA) \neq \bB) = \delta$, then $H(\bB \mid \bA) \leq H_2(\delta)$. 
		\item \label{part:cond-ind} If $\bA$ and $\bB$ are mutually independent conditioned on $\bC$, then $I(\bA ; \bB \mid \bC) = 0$.
		\item \label{part:data-processing} \emph{(Data processing inequality)} $I(f(A) ; B) \leq I(A;B)$. 
	\end{enumerate}
\end{claim}

For two distributions $\distu$ and $\distv$ over the same probability space, the \emph{Kullback-Leibler} divergence between $\distu$ and $\distv$ is defined as 
$\KLD{\distu}{\distv} = \Ex_{a \sim \distu}\left[\log{\frac{\Pr_\distu(a)}{\Pr_\distv(a)}}\right]$. For our proofs, 
we need the following basic relation between KL-divergence and mutual information. 
\begin{claim}\label{clm:it-KL}
	For any two random variables $\bA$ and $\bB$, suppose $\distu$ denotes the distribution of $\bA$ and $\distv_B$ denotes the distribution of $\bA \mid \bB = B$ for any $B$ in the support of $\bB$;
	then, $\Ex_{B}\bracket{\KLD{\distv_B}{\distu}} = I(\bA;\bB)$. 
\end{claim}

We denote the \emph{total variation distance} between two distributions $\distu$ and $\distv$ over the same probability space $\Omega$ by $\card{\distu - \distv} = \frac{1}{2} \cdot \sum_{x \in \Omega}\card{\Pr_\distu(x) - \Pr_{\distv}(x)}$. 
The Pinsker's inequality upper bounds the total variation distance between two distributions based on their KL-divergence as follows.

\begin{claim}[Pinsker's inequality] \label{clm:it-KL-TV} 
		For any two distributions $\distu$ and $\distv$, 
		\[
		\card{\distu - \distv} \leq \sqrt{\frac{1}{2} \cdot \KLD{\distu}{\distv}}
		\]
\end{claim}

Finally, 
\begin{claim} \label{clm:it-TV}
	Suppose $\distu$ and $\distv$ are two probability distributions for a random variable $\bA$ and $J$ is any fixed set in the domain of $\bA$; then $\Pr_{\distu}(\bA \in J) \leq \Pr_{\distv}(\bA \in J) + \card{\distu-\distv}$. 
\end{claim}
\end{document}